\documentclass[
reprint,
superscriptaddress,
amsmath,amssymb,
aps,
prx,
floatfix
]{revtex4-1}

\usepackage{graphicx}% Include figure files
\usepackage{dcolumn}% Align table columns on decimal point
\usepackage{bm}% bold math
\usepackage{dsfont}% double-struck symbols
\usepackage{amsthm}% provides proof environment
\usepackage{chngcntr}
\usepackage{apptools}
\usepackage{amsmath}
\usepackage{mathtools}
\usepackage{todonotes}
\usepackage{graphicx}
\graphicspath{ {./images/} }
\usepackage{siunitx}
\usepackage[margin=1in,footskip=0.25in]{geometry}
\usepackage{blindtext}
\usepackage{enumitem}
\usepackage{xcolor}
\usepackage[makeroom]{cancel}

\usepackage{placeins}
\AtAppendix{\counterwithin{lemma}{section}}
\usepackage{longtable} % two-column tables

\usepackage[capitalise]{cleveref}
\crefname{figure}{Fig.}{Fig.}

\newtheorem{theorem}{Theorem}
\newtheorem{corollary}{Corollary}[theorem]
\newtheorem{lemma}{Lemma}
\theoremstyle{definition}

\theoremstyle{theorem}

\theoremstyle{definition}

\begin{document}

\title{Second-quantized fermionic operators with polylogarithmic\\ qubit and gate complexity}% Force line breaks with \\

\author{William Kirby}
\email{william.kirby@ibm.com}
\affiliation{
Department of Physics and Astronomy, Tufts University, Medford, MA 02155, USA
}

\affiliation{
IBM Quantum, T. J. Watson Research Center, Yorktown Heights, NY 10598, USA
}

\author{Bryce Fuller}
\affiliation{
IBM Quantum, T. J. Watson Research Center, Yorktown Heights, NY 10598, USA
}

\author{Charles Hadfield}
\affiliation{
IBM Quantum, T. J. Watson Research Center, Yorktown Heights, NY 10598, USA
}

\author{Antonio Mezzacapo}
\affiliation{
IBM Quantum, T. J. Watson Research Center, Yorktown Heights, NY 10598, USA
}

\begin{abstract}

We present a method for encoding second-quantized fermionic systems in qubits when the number of fermions is conserved, as in the electronic structure problem.
When the number $F$ of fermions is much smaller than the number $M$ of modes, this symmetry reduces the number of information-theoretically required qubits from $\Theta(M)$ to $O(F\log M)$.
In this limit, our encoding requires $O(F^2\log^4 M)$ qubits, while encoded fermionic creation and annihilation operators have cost $O(F^2\log^5 M)$ in two-qubit gates.
When incorporated into randomized simulation methods, this permits simulating time-evolution with only polylogarithmic explicit dependence on $M$.
This is the first second-quantized encoding of fermions in qubits whose costs in qubits and gates are both polylogarithmic in $M$, which permits studying fermionic systems in the high-accuracy regime of many modes.

\end{abstract}

\pacs{Valid PACS appear here}
\maketitle

\section{Introduction}
\label{intro}

Simulating systems of many interacting fermions is one of the most promising applications for quantum computers.
Many physical systems like molecules, whose accurate simulation would have great practical value, fall into this category.
Classically simulating a many-fermion Hamiltonian is believed to require resources growing exponentially with the system size.
A quantum computer, on the other hand, can simulate time evolution efficiently~\cite{lloyd1996quantumsimulators}, and although ground state problems of interacting Hamiltonians are QMA-complete~\cite{kitaev2002computation}, quantum computers have an exponential memory advantage in representing ground states of many-body systems, compared to classical methods.

Simulating a fermionic system on a quantum computer requires mapping fermionic states and operations to qubit states and operations.
The most well-known methods for accomplishing this are the Jordan-Wigner~\cite{jordanwigner1928} and Bravyi-Kitaev~\cite{bravyikitaev2002,seeley2012bravyikitaev} mappings, which both use one qubit per fermionic mode.
However, the electronic structure Hamiltonian conserves particle number, so one would like to simulate it in the subspace whose number of fermions matches that of the physical system under consideration.
For $F$ fermions in $M$ modes, the theoretical minimum number of qubits required for this is the log of the dimension of the $F$-fermion subspace,
\begin{equation}
\label{inf_thr_opt}
    Q^*=\log_2\begin{pmatrix}M\\F\end{pmatrix}~\xrightarrow{M\gg F}~F\log_2 M.
\end{equation}
The $M\gg F$ limit in \eqref{inf_thr_opt} is exponentially smaller in dependence on $M$ than the number of qubits $Q=M$ of Jordan-Wigner and Bravyi-Kitaev.

The $M\gg F$ limit is an important case in quantum chemistry because it corresponds to simulating a molecule at high accuracy by including many orbitals.
Here ``accuracy" means convergence to the continuum (physical) limit, which is a major but often overlooked problem for quantum simulation of chemistry.
The spin-orbital basis size $M$ required to achieve error $\epsilon$ relative to the continuum limit scales asymptotically as $1/\epsilon$ for reasonable bases~\cite{su2021firstquantization,klopper1995limitingvalues,helgaker1997basisset,halkier1998basisset,harl2008cohesiveenergycurves,hattig2012explicitlycorrelated,kong2012explicitlycorrelated,shepherd2012convergence,gruneis2013planewaves}. 
Hence although existing second-quantized algorithms scale ``efficiently" with $M$, meaning polynomially, the resulting costs are in fact exponential in the number of digits of accuracy relative to the continuum limit.

All prior second-quantized encodings incur these costs in either qubits or gates~\cite{bravyi2017tapering,steudtner2018fermions,steudtner2019fermions,babbush2017cimatrix,kirby2021compactmapping}.
However, for the reasons discussed above, in order to reach high accuracy in quantum chemistry simulations of large molecules, polylogarithmic scalings of both qubits and operations will ultimately become requirements.
In this paper, we present the first second-quantized fermion-to-qubit mapping whose costs in both qubits and operations (two-qubit gates to encode fermionic operators) are polylogarithmic in $M$.
Any quantum simulation algorithm that aspires to have polynomial scaling in the number of digits of accuracy relative to the continuum limit must be based on such a mapping.

\renewcommand{\arraystretch}{1.25}
\begin{table*}[]
    \centering
    \begin{tabular}{l|c|c|c|}
        Citation: & Encoding: & Qubits: & Gates: \\\hline
        Jordan-Wigner~\cite{jordanwigner1928} & Jordan-Wigner & $M$ & $O(M)$ \\
        Bravyi-Kitaev~\cite{bravyikitaev2002,seeley2012bravyikitaev} & Bravyi-Kitaev & $M$ & $O(\log M)$ \\
        Bravyi \emph{et al.}~\cite{bravyi2017tapering} & $Z_2$-symmetries & $M-O(1)$ & $O(M)$ \\
        Bravyi \emph{et al.}~\cite{bravyi2017tapering} & LDPC & $M-\frac{M}{F}$ & $O(M^3)$ \\
        Steudtner-Wehner~\cite{steudtner2018fermions,steudtner2019fermions} & segment & $M-\frac{M}{2F}$ & $O(F^2)$ \\
        Babbush \emph{et al.}~\cite{babbush2017cimatrix} & CI-matrix & $O(F\log M)$ & $O(M)$ \\[0.05in]
        this work & degree-$D$ & $O\left(M^\frac{1}{D+1}DF\log M\right)$ & $O(D^2F^2\log^3M)$ \\[0.07in]
        this work & optimal-degree & $O\left(F^2\log^4M\right)$ & $O\left(F^2\log^5M\right)$
    \end{tabular}
    \caption{Comparison of this work to prior work on encoding second-quantized fermionic Hamiltonians in qubits. ``Gates'' refers to the number of one- and two-qubit gates required to implement the encoding of a conjugate pair of fermionic creation and annihilation operators \eqref{fermionic_conj_pairs}. The exception is the CI-matrix encoding in~\cite{babbush2017cimatrix}, where ``gates" is the cost of the sparse oracle in this sparsity-based approach. The ``degree-$D$ code" in this work is parametrized by a positive integer $D$ that we can choose (subject to certain constraints) and that determines the properties of the code as shown. Choosing $D$ to minimize number of qubits in the $M\gg F$ limit yields the optimal-degree code.}
    \label{prior_work}
\end{table*}

More specifically, for an integer parameter $D$ called the ``degree,'' our encoding requires
\begin{equation}
\label{qubits}
    Q=O\left(M^\frac{1}{D+1}DF\log M\right)~\xrightarrow{M\gg F}~O\left(F^2\log^4M\right)
\end{equation}
qubits, where in the $M\gg F$ limit, $D$ is chosen to minimize the number of qubits.
The cost in two-qubit gates (all controlled phases) and single-qubit gates of an encoded fermionic operator is
\begin{equation}
\label{term_cost_intro}
    O(D^2F^2\log^3M)~\xrightarrow{M\gg F}~O\left(F^2\log^5M\right).
\end{equation}
The cost of implementing a rotation generated by such an operator is the same expression, but in doubly-controlled (i.e., three-qubit) instead of singly-controlled gates.
For comparison, the cost of an oracle query in~\cite{babbush2017cimatrix} is $\Theta(M)$ (which is better than the cost in~\cite{kirby2021compactmapping}; the operation cost for the binary addressing code in~\cite{steudtner2018fermions,steudtner2019fermions} is not analyzed.)
\cref{prior_work} summarizes the comparison of our encoding to prior work.

This paper focuses on encodings of second-quantized fermionic systems, but first-quantized fermion-to-qubit mappings also exist.
Some of these achieve polylogarithmic qubit cost and sublinear gate cost, but only do so for specific basis sets, and require explicit antisymmetrization of the wavefunction~\cite{babbush2019sublinearscaling,su2021firstquantization}.
Our encoding achieves polylogarithmic qubit and gate costs within second-quantization, avoiding these constraints.

Our encoding will be applied in the context of a quantum simulation algorithm.
Many of these have costs that scale polynomially with the number of terms in the Hamiltonian.
For second-quantized electronic structure, the number of terms scales naively as $O(M^4)$, which can sometimes be reduced (e.g.,~\cite{babbush2018lowdepth}) but always remains polynomial in $M$.
Therefore, in a simulation algorithm whose cost is polynomial in the number of terms, the impact of our encoding is reduced because the overall cost of the simulation becomes polynomial in $M$ anyway.

However, some simulation methods based on randomized compiling do not scale explicitly with the number of terms, but instead with the sum $\lambda$ of magnitudes of the Hamiltonian coefficients.
For example, qDRIFT~\cite{campbell2019qdrift} requires $O((\lambda t)^2/\epsilon)$ gates to simulate evolution for a time $t$ with error $\epsilon$, where the gates are rotations generated by terms in the Hamiltonian.
As noted above, in our encoding the cost of implementing such a rotation is given by \eqref{term_cost_intro} in doubly-controlled gates, yielding an overall simulation cost of
\begin{equation}
    O\left(\frac{(\lambda t F)^2}{\epsilon}\log^5 M\right)
\end{equation}
doubly-controlled gates in the $M\gg F$ limit.
The only explicit dependence on $M$ in this formula is polylogarithmic, and this is first quantum simulation algorithm for the electronic structure problem with that property that also only requires polylogarithmically many qubits.
The important caveat to this claim is that the polynomial dependence on the Hamiltonian is now via $\lambda$, and the scaling of $\lambda$ with $M$ is not well-characterized in general for the electronic structure problem.
However, since the coefficients in an electronic structure Hamiltonian vary dramatically in magnitude, scaling with $\lambda$ is much better than scaling with $M$.

Another algorithm well-suited for our encoding is the randomized phase estimation algorithm of~\cite{wan2021phaseestimation}, which uses $\widetilde{O}(1/\eta^2)$ quantum circuits of $\widetilde{O}(\lambda^2/\Delta^2)$ Pauli rotations each to implement phase estimation with additive precision $\Delta$ for a state of overlap at least $\eta$ with the true ground state.
This algorithm is defined for a Hamiltonian decomposed into Pauli operators, but for our encoding the Pauli operators may be replaced by our encoded operators, and the Pauli rotations may be replaced by rotations generated by our encoded operators.
Each circuit in the algorithm will then require
\begin{equation}
    O\left(\frac{(\lambda F)^2}{\Delta^2}\log^5 M\right)
\end{equation}
doubly-controlled gates.

\subsection{Preliminaries}
\label{preliminaries}

We begin with a second-quantized, fermion number conserving Hamiltonian $H$ acting on $M$ modes, i.e., a linear combination of products of creation and annihilation operators $\hat{a}^\dagger_i$ and $\hat{a}_i$.
The example we will bear most strongly in mind is the second-quantized electronic structure Hamiltonian
\begin{equation}
\label{elec_struct}
    H=\sum_{ij}h_{ij}a_i^\dagger a_j+\sum_{ijkl}h_{ijkl}a_i^\dagger a_j^\dagger a_ka_l,
\end{equation}
where the indices $i,j,k,l$ run over the modes (spin-orbitals), and the coefficients $h_{ij}$ and $h_{ijkl}$ are the one- and two-body integrals, respectively~\cite{mcardle2020quantumchem}.
The electronic structure Hamiltonian is of particular interest because, in addition to being of great importance in computational chemistry, in this case the $M\gg F$ limit corresponds to studying a fixed molecule (with a fixed number of electrons) in the high precision limit (many modes), as discussed above.

The Bravyi-Kitaev (BK) transformation~\cite{bravyikitaev2002,seeley2012bravyikitaev} maps fermionic states and operators to qubit states and operators such that the conjugate pairs
\begin{equation}
\label{fermionic_conj_pairs}
    \hat{a}^\dagger_i+\hat{a}_i\quad\text{and}\quad i(\hat{a}^\dagger_i-\hat{a}_i)
\end{equation}
(i.e., Majorana operators) are mapped to Pauli operators containing $O(\log M)$ nonidentity single qubit Pauli matrices (see~\cite[eq. (39-40)]{seeley2012bravyikitaev} as well as more detailed discussion in \cref{encoding_operations}).
Although individual creation and annihilation operators are neither unitary nor Hermitian, the conjugate pairs \eqref{fermionic_conj_pairs} are both, and the Hamiltonian \eqref{elec_struct} may be rewritten as a linear combination of products of these~\cite{bravyikitaev2002}.
Hence, under the BK mapping the Hamiltonian becomes a linear combination of Pauli operators:
\begin{equation}
\label{H_init}
    H=\sum_{P\in\mathcal{P}^{\otimes M}}h_PP,
\end{equation}
where $\mathcal{P}$ is the set of single-qubit Pauli matrices and identity, and the $h_P$ are real coefficients.

We want to simulate this Hamiltonian within the $F$-fermion subspace.
In the BK mapping, each occupation number state is represented as a bitstring $b$ whose entries correspond to parities of subsets of the fermionic modes.
We will refer to these as \emph{BK bitstrings}.
For a single-fermion state (i.e., an occupation number state in which a single mode is occupied), the corresponding BK bitstring contains at most $\left\lceil\log_2(M+1)\right\rceil$ $1$s.
The BK mapping is linear, so the BK bitstring corresponding to a multi-fermion state is the bitwise sum of the single-fermion BK bitstrings corresponding to occupied modes.
Hence, a BK bitstring $b$ corresponding to an occupation number state of $F$ fermions has Hamming weight at most $F\left\lceil\log_2(M+1)\right\rceil$, which we denote
\begin{equation}
\label{Hamming_weight_def}
    |b|\le F\left\lceil\log_2(M+1)\right\rceil\equiv G,
\end{equation}
i.e., $b$ contains at most $G=F\left\lceil\log_2(M+1)\right\rceil$ $1$s.
While in the Jordan-Wigner encoding, the Hamming weight is exactly equal to $F$, the Bravyi-Kitaev Hamming weight bound \eqref{Hamming_weight_def} will be sufficient for us to exploit fermion number conservation, and indeed our encoding will apply to any bitstrings up to and including Hamming weight $G$.
Although the BK mapping is typically used to map a fermionic Hamiltonian to a qubit Hamiltonian, we will think of $H$ in \eqref{H_init} as the unencoded Hamiltonian that will be the starting point for our encoding.

\section{Encoding states}
\label{encoding_states}

We will encode the Hamiltonian $H$ in \eqref{H_init} in a qubit Hamiltonian that acts on $Q<M$ qubits.
The encoding $\mathcal{E}$ will satisfy several properties:
\begin{enumerate}
    \item $\mathcal{E}$ maps occupation number states containing up to $F$ fermions, i.e., BK bitstrings of Hamming weight up to $G$, to qubit computational basis states.
    \item $\mathcal{E}$ is linear on bitwise addition $\oplus$ of bitstrings (bitwise XOR), i.e.,
    \begin{equation}
    \label{linear}
        \mathcal{E}(a\oplus b)=\mathcal{E}(a)\oplus\mathcal{E}(b)
    \end{equation}
    for two BK bitstrings $a,b$.
    \item The $i$th bit in $b$ ($b_i$) is associated to a set $S_i$ of qubits such that for an up to $F$-fermion state, $b_i=1$ if and only if in the encoded state more than half of the qubits in $S_i$ are $1$.
    \item $\mathcal{E}$ is invertible for occupation number states containing up to $F$ fermions, i.e., BK bitstrings $b$ with $|b|\le G$. (This follows from property 3.)
\end{enumerate}
Since the occupation number states form a basis for the fermionic Hilbert space, properties 1, 2, and 4 imply that the map $\mathcal{E}$ extends to an invertible linear transformation sending the space of $F$-fermion wavefunctions into a subspace of the $Q$-qubit Hilbert space.
We will call $\mathcal{E}(b)$ the \emph{codeword} for $b$, where $b$ is a BK bitstring.
The span of the codewords will be called the \emph{codespace}, and not every qubit computational basis state must be a codeword, so the codespace is not necessarily the entire $Q$-qubit Hilbert space.

In the next section, we will use the third property above to construct efficient implementations of encoded fermionic operators.
The first two properties imply that the encoding is defined by its action on BK bitstrings with Hamming weight one, which we call \emph{elementary bitstrings}.
By linearity, if we specify the encodings of these, which we call \emph{elementary codewords}, then the encoding of any higher-weight BK bitstring $b$ is the bitwise sum of the elementary codewords corresponding to the $1$s in $b$.
Hence, the fourth property (invertibility) will hold if and only if bitwise sums of up to $G$ of the elementary codewords are unique.

To guarantee that properties three and four above hold, the encoding we construct will satisfy the following sufficient conditions: if $\alpha=\mathcal{E}(a)$ and $\beta=\mathcal{E}(b)$ are elementary codewords for different elementary bitstrings $a$ and $b$, then
\begin{equation}
\label{code_properties}
\begin{split}
    &\text{len}(\alpha)=\text{len}(\beta)=Q,\\
    &|\alpha|=|\beta|=L,\\
    &\alpha\cdot\beta=\sum_{i=0}^{Q-1}\alpha_i\beta_i\le D<\frac{L}{2G},
\end{split}
\end{equation}
where $D$ is some maximum allowed overlap of the codewords.
If $b$ is the elementary bitstring in which (only) bit $b_i=1$, the $L$ $1$s in $\mathcal{E}(b)$ are exactly the qubits in $S_i$.
Sets that satisfy \eqref{code_properties} also satisfy properties three and four of $\mathcal{E}$, above, which we prove as \cref{overlap_lemma} in \cref{proofs}.

Having established the properties that $\mathcal{E}$ must satisfy, we can now specify $\mathcal{E}$ by constructing the elementary codewords.
For fixed $G=F\left\lceil\log_2(M+1)\right\rceil$, $\mathcal{E}$ is parametrized by positive integers $L$ and $D$ satisfying \eqref{code_properties}.
For a given $D$, we will later want $L$ to be as small as possible, so we will choose
\begin{equation}
\label{L_min}
    L=2DG+1.
\end{equation}

Let the range of $\mathcal{E}$ be the computational basis states of $Q=L'L$ qubits, which are partitioned into $L$ blocks of $L'$ qubits where $L'$ is a prime number lower bounded by $L$.
In each block, one of the qubits will be $1$ and the others will be $0$, so $L$ qubits in total are $1$, i.e., the elementary codewords have Hamming weight $L$, as required by \eqref{code_properties}.
Each elementary codeword is thus equivalent to a function $y:\mathbb{Z}_L\rightarrow\mathbb{Z}_{L'}$ (where $\mathbb{Z}_n$ denotes the ring of integers modulo $n$), which maps the index $x$ of a block to the position $y(x)$ of the $1$ in that block.
Examples of this mapping are given in Figs.~\ref{encoding_examples} and~\ref{graph_linearization_example}.

\begin{figure}[t]
    \centering
    \begin{equation*}
    \begin{split}
        y(x)=0\quad\leftrightarrow\quad&10000\,10000\,10000\,10000\,10000\\
        y(x)=x\quad\leftrightarrow\quad&10000\,01000\,00100\,00010\,00001\\
        y(x)=2+x\quad\leftrightarrow\quad&00100\,00010\,00001\,10000\,01000\\
        y(x)=x^2\quad\leftrightarrow\quad&10000\,01000\,00001\,00001\,01000\\
    \end{split}
    \end{equation*}
    \caption{Examples of the correspondence between functions from $\mathbb{Z}_L$ to $\mathbb{Z}_{L'}$ and elementary codewords. We have inserted spaces between the $L=5$ blocks of $L'=5$ qubits. $y(x)$ is the index of the $1$ in the $x$th block of qubits.}
    \label{encoding_examples}
\end{figure}

\begin{figure}
    \centering
    \includegraphics[width=\columnwidth]{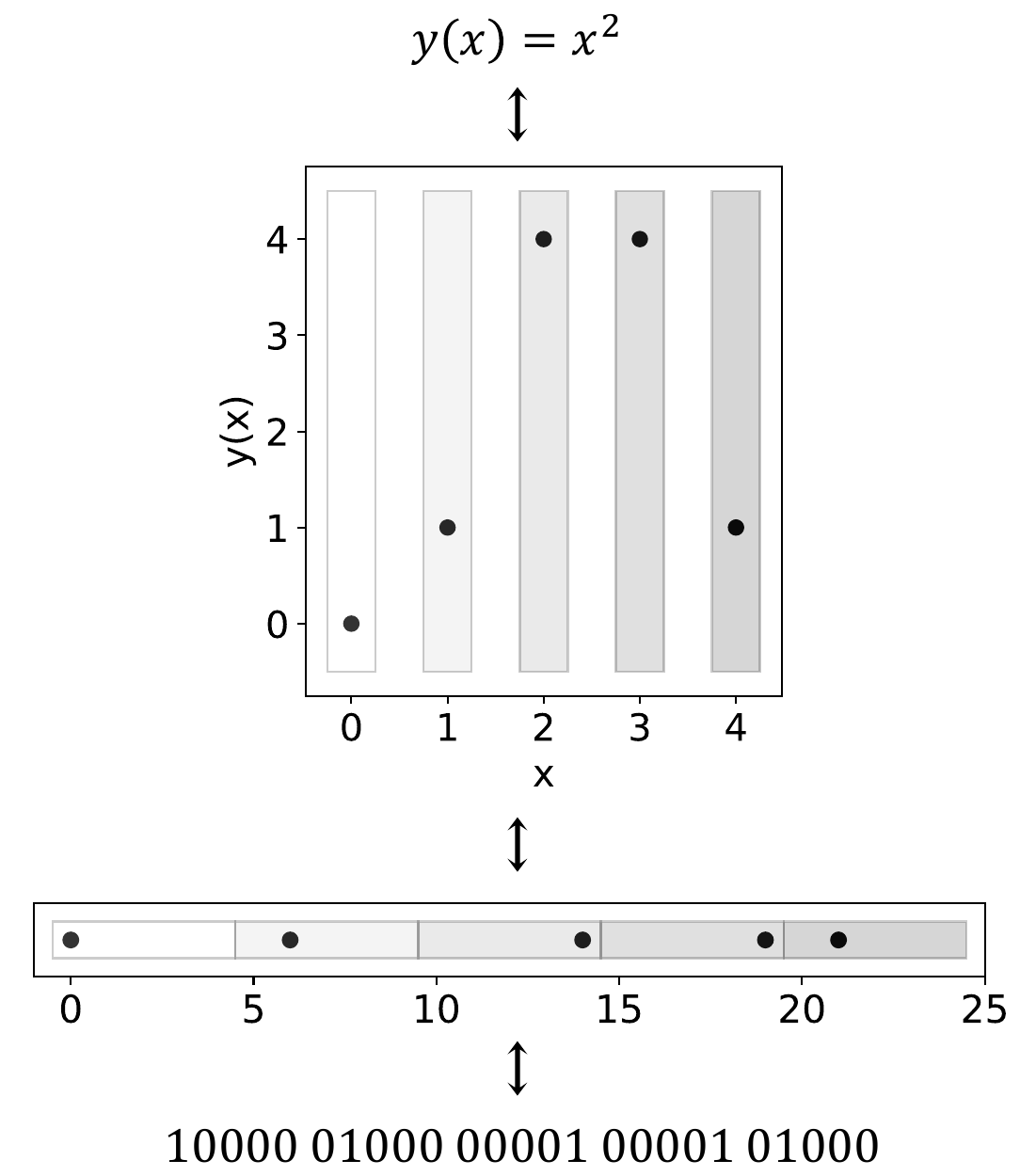}
    \caption{The mapping from a function over $\mathbb{Z}_{L'}$ to an elementary codeword may be viewed as a linearization of the graph of the function, as illustrated in this figure for the function $y(x)=x^2$ over $\mathbb{Z}_5$. To obtain the linearization, the columns of the graph of $y(x)$ are laid out in order horizontally: the locations of the points in the linearization give the locations of the $1$s in the corresponding codeword. The shading of the columns in the graph matches their shading in its linearization.}
    \label{graph_linearization_example}
\end{figure}

We want $D$ to be an upper bound on the overlaps of the elementary codewords, as in \eqref{code_properties}.
Since for any pair $\alpha,\beta$ of elementary codewords, each contains a single $1$ in each block of qubits and the corresponding functions $y_\alpha$ and $y_\beta$ give the locations of the $1$s, this is equivalent to $D$ being an upper bound on the number of intersections of $y_\alpha$ and $y_\beta$.
Therefore, let $y_\alpha$ and $y_\beta$ be distinct degree-$D$ polynomials over $\mathbb{Z}_{L'}$, with their domains restricted to $\mathbb{Z}_L\subseteq\mathbb{Z}_{L'}$.
In this case, their difference is also a polynomial of degree at most $D$, so it can have at most $D$ roots.
Hence $y_\alpha$ and $y_\beta$ can intersect in at most $D$ points, and thus the corresponding codewords can overlap in at most $D$ bits, as desired.
Technically, $y_\alpha$ and $y_\beta$ are \emph{polynomial functions} (as opposed to formal algebraic polynomials), but here we refer to them as polynomials for simplicity; see \cref{finite_fields_app} for details and a review of the properties of polynomials over finite fields.

It follows that if each of our elementary codewords corresponds to a distinct degree-$D$ polynomial over $\mathbb{Z}_{L'}$ as described above, its overlap with any other elementary codeword is upper bounded by $D$.
This still holds if we include all polynomials of degree at most $D$.
There are $(L')^{D+1}$ distinct polynomials of degree at most $D$ over $\mathbb{Z}_{L'}$, since each polynomial is uniquely specified by its coefficients, and each of the $D+1$ coefficients of $x^0,x^1,...,x^D$ is one of the $L'$ elements of $\mathbb{Z}_{L'}$ (this relies on the fact that $\mathbb{Z}_{L'}$ has prime order and that $D<L'$; see \cref{finite_fields_app}).
We encode one Bravyi-Kitaev bit in each of the corresponding codewords, and the number of modes is equal to the number of Bravyi-Kitaev bits, so we encode $(L')^{D+1}$ modes in $L'L$ qubits.

Hence as long as $D\ge2$, or $L'>L$ and $D\ge1$, this encoding permits $Q<M$.
The $D=0$ case reduces to the Bravyi-Kitaev encoding: by \eqref{L_min}, the elementary codewords have Hamming weight $L=1$, so there is a single block of $L'$ qubits, and the $L'$ degree-$0$ polynomials (constants) simply give the possible locations of the single $1$.

For generic values of $L$, $L'$, $M$, and $D$, we could partition our qubits into subsets of size $L'L$, and use each subset to encode $(L')^{D+1}$ modes as described above: this would require
\begin{equation}
\label{qubits_exact_partition}
    Q=\left\lceil\frac{M}{(L')^{D+1}}\right\rceil L'L
\end{equation}
qubits.
However, from \eqref{qubits_exact_partition} we can see that for fixed $D$ (and thus $L$), it is in fact best to use only one such subset, in which case we must choose $L'$ to be the least prime such that
\begin{equation}
\label{Lprime}
    (L')^{D+1}\ge M
\end{equation}
(provided $L'\ge L$).
This is the minimum value of $L'$ such that all of the modes are encoded in a single set of $L'L$ qubits, so choosing $L'$ larger than this would be disadvantageous.
This value of $L'$ yields the number of qubits required to encode $F$ fermions in $M$ modes via the degree-$D$ code: by \eqref{L_min}, \eqref{qubits_exact_partition}, and \eqref{Lprime},
\begin{equation}
\label{qubits_exact}
\begin{split}
    Q&=L'L\\
    &=M^\frac{1}{D+1}(2DF\left\lceil\log_2(M+1)\right\rceil+1)+O\left(F\log^2 M\right)
\end{split}
\end{equation}
on average, since by the prime number theorem, the least prime greater than $M^\frac{1}{D+1}$ exceeds it by $O\left(\log\left(M^\frac{1}{D+1}\right)\right)=O\left(\frac{\log M}{D}\right)$ on average.
Although the above is an average-case statement, by the Bertrand-Chebyshev Theorem, the least prime greater than $M^\frac{1}{D+1}$ is upper-bounded by $2M^\frac{1}{D+1}$, so $Q$ can never be worse than twice the first term in the second line of \eqref{qubits_exact}.

The $M\gg F$ limit of our encoding is an important case in practice, as discussed in the introduction.
The ideal number of qubits in this limit is given in \eqref{inf_thr_opt}. The performance of our code in this limit is given by the following theorem:
\begin{theorem}
\label{performance_theorem}
    In the $M\gg F$ limit, our code satisfies
    \begin{equation}
        Q=O\left(F^2\log^4 M\right),
    \end{equation}
    with $D$ satisfying $D=O(\log M)$.
\end{theorem}
\noindent
The proof can be found in \cref{proofs}.

Finally, as discussed above, $D=1$ is the least degree for which our encoding can be advantageous over the Bravyi-Kitaev encoding.
In this case, $L'L=\Theta(\sqrt{M}F\log M)$, so our code is asymptotically advantageous when $F=O(M^{1/2-\epsilon})$ for $\epsilon>0$.
As a non-asymptotic example, consider a water molecule, which contains $F=10$ electrons: in this case our code becomes advantageous over the Bravyi-Kitaev encoding when ${M\ge118328}$.
Beyond this point the qubit cost for our code grows much more slowly than $M$: for example, when $M=10^6$ the optimal value of $D$ is still $1$, and our code requires $\sim4\times10^5$ qubits, and when $M=10^7$ the optimal value of $D$ is $2$ and our code requires $\sim9\times10^5$ qubits.
The exact point at which our code becomes preferable over other options in general is discussed in \cref{threshold_app}.
For smaller $M$, we recommend using the ``segment code'' of \cite{steudtner2018fermions,steudtner2019fermions} (see \cref{prior_work}), for which operations can be implemented efficiently using the construction in the proof of \cref{parity_theorem}.
This is discussed in \cref{application_to_segment_code}.
The segment code becomes advantageous over Bravyi-Kitaev when $F\le\frac{M}{2}-1$ and yields $Q\approx\left(1-\frac{1}{2F}\right)M$, so it bridges the gap to the large-$M$ regime where our encoding becomes preferable.

\section{Encoding operations}
\label{encoding_operations}

In the Bravyi-Kitaev mapping, a conjugate pair of fermionic operators as in \eqref{fermionic_conj_pairs} is mapped to a Pauli operator with nonidentity action on $O(\log M)$ qubits, i.e., bits in the BK bitstring $b$~\cite{bravyikitaev2002,seeley2012bravyikitaev}.
Up to a phase $\pm1$ or $\pm i$, each such Pauli operator can be expressed as a product of $O(\log M)$ single-qubit Pauli operators $X$ and $Z$.
Let us denote these as $X^{(BK)}_i$ and $Z^{(BK)}_i$, where $i$ indexes the bit $b_i$ they act upon; operators on the codespace will be written with no superscript.

If we can implement the encodings of $X^{(BK)}_i$ and $Z^{(BK)}_i$ as unitaries on the codespace, we can implement any term in the Hamiltonian \eqref{H_init} as a unitary operator.
This means that we can implement the Hamiltonian as a linear combination of unitaries and simulate time-evolution~\cite{childs2012lcu,berry2015truncatedtaylor,berry2015nearlyoptimal,low2017signalprocessing,low2018interactionpicture,low2019qubitization,berry2020l1norm,campbell2019qdrift,wan2021phaseestimation}, with the randomized algorithms of~\cite{campbell2019qdrift,wan2021phaseestimation} most likely being the best choices for our encoding as discussed in the introduction. Alternatively, we can estimate the expectation value of each term via Hadamard tests and implement a variational quantum eigensolver (VQE) that searches for the Hamiltonian's ground state energy~\cite{peruzzo2014vqe,kirby2021sparsevqe}.

Each bit $i$ in the BK bitstring is associated to some set $S_i$ containing the indices of the $L$ qubits that are $1$ in the corresponding elementary codeword (see \eqref{code_properties} and the corresponding discussion, above).
Hence, because our encoding is linear \eqref{linear}, the encoding of $X^{(BK)}_i$ is
\begin{equation}
\label{X_mapping}
    \mathcal{E}\Big(X^{(BK)}_i\Big)=\prod_{j\in S_i}X_j,
\end{equation}
i.e., bitflips on all bits that are 1 in the elementary codeword corresponding to $i$ (we abuse the notation $\mathcal{E}$ to denote the encoding of operators as well as of states).

To implement the encoding of $Z^{(BK)}_i$, we use the fact that unencoded bit $b_i=1$ if and only if more than half of the code qubits in $S_i$ are $1$ (property 3 in~\cref{encoding_states}).
Any computational basis state $|q\rangle$ of the code qubits is an eigenstate of $\sum_{j\in S_i}Z_j$, and by the previous sentence, the eigenvalue is negative if and only if $b_i=1$.
This means that, for integer $z_i$ defined by
\begin{equation}
\label{Z_sum}
    \sum_{j\in S_i}Z_j|q\rangle=z_i|q\rangle,
\end{equation}
we have
\begin{equation}
\label{Z_action}
    \mathcal{E}\Big(Z^{(BK)}_i\Big)|q\rangle=
    \begin{cases}
        |q\rangle\quad\text{if $z_i>0$},\\
        -|q\rangle\quad\text{if $z_i<0$}.
    \end{cases}
\end{equation}
Note that since $|S_i|$ is odd, $z_i\neq0$.

Hence, we just need to implement the ``majority-vote" operation given by \eqref{Z_action} for any set $S_i$ of $L$ qubits.
We can accomplish this by observing that $\sum_{j\in S_i}Z_j$ has only $L+1$ distinct eigenvalues, so we can use Hermite interpolation~\cite{burden2015numerical} to efficiently express the desired operation as a polynomial of $\sum_{j\in S_i}Z_j$ (really of a rescaling of $\sum_{j\in S_i}Z_j$).
We can then exactly implement this polynomial using quantum signal processing~\cite{low2017signalprocessing,low2019qubitization}.
The details are given in the proof of the following theorem:
\begin{theorem}
\label{parity_theorem}
    The operation $\mathcal{E}\Big(Z^{(BK)}_i\Big)$ defined by \eqref{Z_action} can be implemented using
    \begin{equation}
    \label{parity_cost}
    \begin{split}
        L(2L-1)=&~8D^2F^2\lceil\log_2(M+1)\rceil^2\\
        &~+6DF\lceil\log_2(M+1)\rceil+1\\
        =&~O(D^2F^2\log^2M)
    \end{split}
    \end{equation}
    controlled phases and single-qubit gates, and one ancilla qubit.
\end{theorem}

\noindent
\begin{proof}
We want to implement the encoded parity operator $\mathcal{E}\Big(Z^{(BK)}_i\Big)$ whose action on qubits is given by \eqref{Z_action}.
To do this, we can use quantum signal processing~\cite{low2017signalprocessing,low2019qubitization,low2016resonantequiangular}.
First, define the Hermitian operator
\begin{equation}
\label{Hm_def}
    \mathcal{H}_i\equiv
    \cos\mathcal{G}_i\equiv\cos\left(\frac{\pi}{2}\left(\mathds{1}-\frac{1}{L}\sum_{j\in S_i}Z_j\right)\right),
\end{equation}
for
\begin{equation}
\label{Gm_def}
    \mathcal{G}_i\equiv\frac{\pi}{2}\left(\mathds{1}-\frac{1}{L}\sum_{j\in S_i}Z_j\right).
\end{equation}
By definition, $\mathcal{H}_i$ has eigenvalues
\begin{equation}
    \{\cos\left(\frac{m\pi}{L}\right)~|~m=0,1,2,...,L\}.
\end{equation}
Any computational basis state $|q\rangle$ is an eigenvector of $\mathcal{H}_i$, so if we let
\begin{equation}
    \mathcal{H}_i|q\rangle=\lambda_q|q\rangle,
\end{equation}
then by \eqref{Z_sum}, \eqref{Z_action}, and \eqref{Hm_def},
\begin{equation}
\label{Zm_in_terms_of_Hm}
    \mathcal{E}\Big(Z^{(BK)}_i\Big)|q\rangle=
    \begin{cases}
        |q\rangle\quad\text{if $\lambda_q>0$},\\
        -|q\rangle\quad\text{if $\lambda_q<0$}.
    \end{cases}
\end{equation}

Next, we define a block-encoding $W_\phi$ of $\mathcal{H}_i$ ($W_\phi$ is the \emph{phased iterate} of~\cite{low2019qubitization}):
\begin{equation}
\label{phased_iterate_def}
    W_\phi\equiv
    \begin{pmatrix}
        \mathcal{H}_i&-ie^{-i\phi}\sqrt{1-\mathcal{H}_i^2}\\
        -ie^{i\phi}\sqrt{1-\mathcal{H}_i^2}&\mathcal{H}_i
    \end{pmatrix},
\end{equation}
which acts on the codespace and one additional ancilla qubit whose states $|\cdot\rangle_a$ define the blocks in \eqref{phased_iterate_def}.
Using quantum signal processing, via $N$ queries to $W_\phi$ we can implement
\begin{equation}
\label{block_encoding}
    \begin{pmatrix}
        A(\mathcal{H}_i)&\cdot\\
        \cdot&\cdot
    \end{pmatrix}
\end{equation}
for any degree-$N$ real polynomial $A$ such that
\begin{equation}
\label{poly_constraints}
\begin{split}
    &|A(\lambda)|\le1\quad\forall\lambda\in[-1,1],\\
    &|A(\lambda)|\ge1\quad\forall\lambda\notin(-1,1),
\end{split}
\end{equation}
by~\cite[Lemma 12]{low2019qubitization} (for us, $N$ will always be odd, so the final condition in~\cite[Lemma 12]{low2019qubitization} is irrelevant).
Hence, we just want to find a polynomial $A$ that satisfies the above properties and passes through the points
\begin{equation}
\label{hermite_points}
\begin{split}
    &\{\left(\cos\left(\frac{m\pi}{L}\right),1\right)~|~m=0,1,2,...,\left\lfloor\frac{L}{2}\right\rfloor\}\\
    &\cup\{\left(\cos\left(\frac{m\pi}{L}\right),-1\right)~|~m=\left\lfloor\frac{L}{2}\right\rfloor+1,...,L\},
\end{split}
\end{equation}
since this will give
\begin{equation}
    \mathcal{E}\Big(Z^{(BK)}_i\Big)|q\rangle=A(\mathcal{H}_i)|q\rangle
\end{equation}
for any computational basis state $|q\rangle$ by \eqref{Zm_in_terms_of_Hm}, and thus for all qubit states, including the codespace.
Note that \eqref{hermite_points} finally justifies why we require $L$ to be odd: this guarantees that the numbers of points with value $+1$ and value $-1$ are the same.

\begin{figure}[t]
    \centering
    \includegraphics[width=\columnwidth]{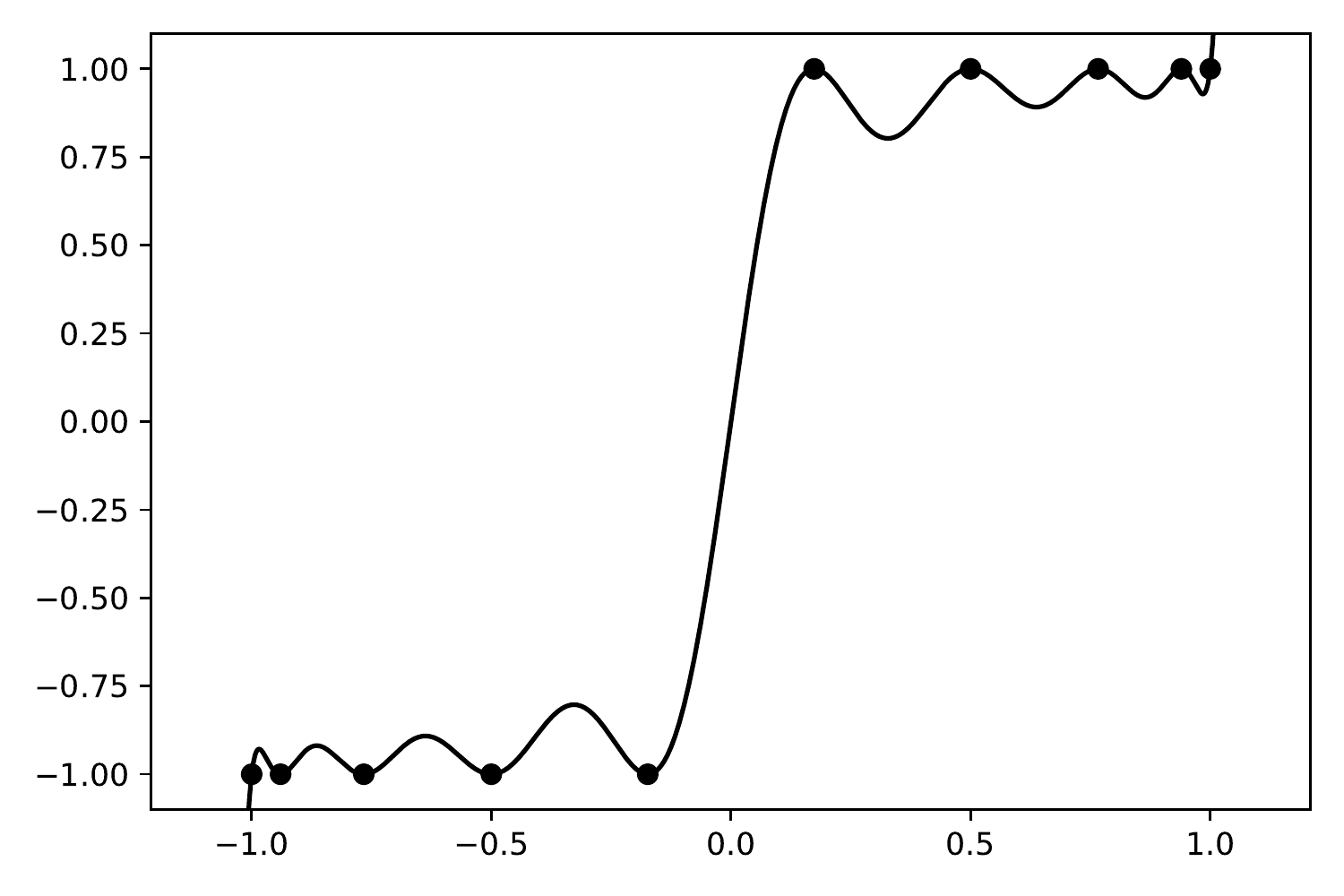}
    \caption{Example of the Hermite interpolating polynomial for $L=9$. The points are given by \eqref{hermite_points}, and all of their derivatives are set to zero except for the leftmost and rightmost points.}
    \label{hermite_example}
\end{figure}

We can find such a polynomial by Hermite interpolation of the points \eqref{hermite_points} together with the constraints that the first derivative be zero at each non-edge point (i.e., all points except for $(1,1)$ and $(-1,-1)$, the leftmost and rightmost points).
The constraints on the derivatives are necessary for the resulting polynomial to satisfy the first line in \eqref{poly_constraints}.
See \cref{hermite_example} for an example, and \cref{hermite_interpolation_app} for a review of Hermite interpolation.

We prove as \cref{interpolation_lemma}, below, that the resulting polynomial satisfies the constraints \eqref{poly_constraints} except for
\begin{equation}
\label{local_minima}
\begin{split}
    &A(\lambda)\ge-1\quad\forall\lambda\in[0,1],\\
    &A(\lambda)\le1\quad\forall\lambda\in[-1,0],
\end{split}
\end{equation}
e.g., the local minima in the right half of \cref{hermite_example} do not go below $-1$ and the local maxima in the left half of \cref{hermite_example} do not go above $1$.
Confirming these for general $L$ proved difficult because of the arbitrary degree of the polynomial.
However, we checked \eqref{local_minima} for all (odd) $L$ up to $501$, and found that for all these, the least local minimum in the $\lambda>0$ region is greater than $0.8$ and increases slightly with $L$: from $L=251$ to $L=501$, the first five digits of this least local minimum are $0.80662$, while the remaining digits climb slowly.
Correspondingly, the greatest local maximum in the $\lambda<0$ region is less than $-0.8$ and decreases slightly with $L$.
Since $L=501$ corresponds to at least $LL'\ge L^2=251001$ qubits, we consider this result to be adequate for intermediate-term applications, and we conjecture that \eqref{local_minima} holds for all (odd) $L$.

Since $\mathcal{E}\Big(Z^{(BK)}_i\Big)$ is unitary and its block-encoding via $A(\mathcal{H}_i)$ as in \eqref{block_encoding} must also be unitary, the resulting block encoding must be
\begin{equation}
    \begin{pmatrix}
        A(\mathcal{H}_i)&0\\
        0&\cdot
    \end{pmatrix}
    =
    \begin{pmatrix}
        \mathcal{E}\Big(Z^{(BK)}_i\Big)&0\\
        0&\cdot
    \end{pmatrix},
\end{equation}
i.e., there is no leakage out of the upper left block.
Hence if we start with a state
\begin{equation}
    |0\rangle_a\otimes|q\rangle,
\end{equation}
the quantum signal processing algorithm will map this exactly to
\begin{equation}
    |0\rangle_a\otimes\mathcal{E}\Big(Z^{(BK)}_i\Big)|q\rangle
\end{equation}
as desired.
Hermite interpolation of $L+1$ points and $L-1$ first derivatives results in a degree
\begin{equation}
    N=2L-1
\end{equation}
polynomial, so this algorithm requires $2L-1$ queries to $W_\phi$.

It remains to show how to implement $W_\phi$.
\cref{phased_iterate_implementation_lemma}, below, shows how this can be done using $L$ two-qubit operations, the controlled-$Z$-rotations in \eqref{Wphi_decomp_main}.
Since the algorithm requires $2L-1$ queries to $W_\phi$, the total number of two-qubit operations required to implement $Z^{(BK)}_i$ is $L(2L-1)$ as claimed in the theorem statement.
The number of additional single-qubit gates required is equal to this plus $5(2L-1)$ (for the single-qubit gates acting on the ancilla), by \eqref{Wphi_decomp_main}.
This completes the proof of \cref{parity_theorem}.
\end{proof}

As pointed out in the proof above, \cref{parity_theorem} relies on a conjecture that we have checked explicitly out to at least 251001 qubits.
It also relies on the two lemmas that follow, whose proofs we leave for \cref{proofs}:

\begin{lemma}
\label{interpolation_lemma}
    For odd $L$, the polynomial $A$ obtained by Hermite interpolation of the points \eqref{hermite_points}, together with the constraints that its first derivative be zero at all non-edge points, satisfies:
    \begin{equation}
        \label{poly_constraints_proof}
        \begin{split}
            &A(\lambda)\le1\quad\forall\lambda\in[0,1],\\
            &A(\lambda)\ge-1\quad\forall\lambda\in[-1,0],\\
            &|A(\lambda)|\ge1\quad\forall\lambda\notin(-1,1).
        \end{split}
    \end{equation}
    The $x$-coordinates of the non-edge points are
    \begin{equation}
        \{\cos\left(\frac{m\pi}{L}\right)~|~m=1,2,...,L-2,L-1\}.
    \end{equation}
\end{lemma}

\begin{lemma}
\label{phased_iterate_implementation_lemma}
    For
    \begin{equation}
    \label{R_def_main}
        R_\phi\equiv
        \begin{pmatrix}
            1&0\\
            0&e^{i\phi}
        \end{pmatrix}_a
        \otimes\mathds{1},
    \end{equation}
    $W_\phi$ is given by the following sequence of operations:
    \begin{equation}
    \label{Wphi_decomp_main}
    \begin{split}
        W_{\phi}=&R_\phi(H\otimes\mathds{1})\\
        &\cdot R_\pi\left(\prod_{j\in S_i}\text{ctrl-}e^{-\frac{i\pi Z_j}{L}}\right)\left(\prod_{j\in S_i}e^{-\frac{i\pi Z_j}{2L}}\right)\\
        &\cdot(H\otimes\mathds{1})R_\phi^\dagger,
    \end{split}
    \end{equation}
    where $H$ is the Hadamard gate, $Z_j$ is a single-qubit Pauli-$Z$ acting on code qubit $j$, and the controls are on the ancilla qubit.
\end{lemma}

It follows from \cref{parity_theorem} that the cost in two-qubit gates of implementing an encoded conjugate pair of fermionic operators is $O(\log M)$ times \eqref{parity_cost}:
\begin{equation}
\label{term_cost}
    O\left(D^2F^2\log^3 M\right)~\xrightarrow{M\gg F\text{ limit}}~O\left(F^2\log^5 M\right),
\end{equation}
where the $M\gg F$ limit is obtained by substituting $D=O(\log M)$, per \cref{performance_theorem}.
Since every term in the Hamiltonian is a product of up to four such conjugate pairs, the cost of implementing the encoding of a term in the Hamiltonian as a unitary also scales as \eqref{term_cost}.

To implement VQE we need to construct encoded fermion number preserving ansatz circuits, and to implement many simulation methods including the randomized methods of~\cite{campbell2019qdrift,wan2021phaseestimation} we require encodings of rotations generated by terms in the Hamiltonian.
In \cref{proofs}, we prove the following corollary to \cref{parity_theorem}:
\begin{corollary}
\label{ansatz_implementation_lemma}
    We can implement an encoded \emph{hop gate}, which is universal for real-valued wavefunctions with fixed fermion number~\cite{eddins2021entanglementforging}, or the encoding of a rotation generated by any term in the Hamiltonian, using one ancilla qubit and a number of doubly-controlled (three-qubit) gates given by \eqref{term_cost}.
\end{corollary}
\noindent
Being able to construct an encoded hop gate means that in principle we can implement any desired fermion number preserving ansatz circuit, augmenting with single-fermion phases if complex-valued wavefunctions are desired.
The hop gate in practice is best suited for constructing so-called ``hardware-efficient" ansatze~
\cite{eddins2021entanglementforging}.
In terms of ansatze, encoded rotations generated by terms in the Hamiltonian are aimed more specifically at implementing the unitary coupled-cluster ansatz~\cite{romero18a}.
However, in either case the main challenge for executing these encoded gates in practice is coherently implementing sequences of doubly-controlled phases and single-qubit gates.
Since circuit depths on existing quantum computers are severely limited by noise, our method is mainly targeted at future devices where this noise is reduced through hardware improvements and mitigation, or where error-correction is possible.

Furthermore, \cref{parity_theorem} assumes arbitrary connectivity.
A fixed qubit architecture requires additional qubit swaps.
In \cref{proofs}, we prove as \cref{fixed_lattice_lemma} that the number of swaps required to implement $\mathcal{E}\Big(Z^{(BK)}_i\Big)$ on a linear qubit architecture is $O(QL)$, which is still polylogarithmic in $M$.
This therefore also holds on any architecture (such as planar architecture) that includes linear connectivity as a subgraph.

Finally, although the Hamiltonian \eqref{elec_struct} conserves fermion number, its individual unitary terms after transforming it into a linear combination of Majorana operators may increase the number of fermions by at most four, since they are products of at most four Majorana operators \eqref{fermionic_conj_pairs}.
The contributions from different terms to states with extra fermions must cancel out so that the whole Hamiltonian does conserve fermion number, but in order for them to cancel in the encoded Hamiltonian we must ensure that states of up to $F+4$ fermions are correctly encoded.
Hence, $F$ should be replaced by $F+4$ in all of our costs for both qubits and gates, but this only changes the scalings at subleading order.

\section{Conclusion}
\label{conclusion}

In this paper, we presented the first second-quantized fermion-to-qubit mapping that uses polylogarithmically-many qubits and gates in the number of fermionic modes, to simulate fermionic creation and annihilation operators.
This is an exponential improvement in the dependence on number of modes compared to prior second-quantized encodings, for either qubits, operations, or both.
Polylogarithmic dependence on the number of modes will permit simulation of molecules as well as many-body problems such as the Hubbard model in the high-accuracy limit of large bases.

The method of using quantum signal processing to exactly implement the encoded parity operation as in \cref{parity_theorem} may have utility beyond the scope of this paper.
It permits implementing a $\pm1$ phase controlled on the Hamming weight of a set of $n$ qubits, provided the corresponding Hermite interpolating polynomial satisfies \eqref{poly_constraints}.
For example, it could be used to implement an $n$-qubit controlled phase, by controlling the $-1$ phase on the Hamming weight of the qubits being $n$.
Just as in \cref{parity_theorem}, the cost of implementing this via Hermite interpolation and quantum signal processing is $n(2n-1)$ singly-controlled phases plus $O(n)$ single-qubit gates.
The proof of \cref{interpolation_lemma} can be adapted to show that the corresponding interpolation polynomial satisfies \eqref{poly_constraints} up to a conjecture similar to that involved in \cref{parity_theorem}, using the same phased iterate as in \cref{parity_theorem}.
We checked the conjecture out to $n=397$ and found that the least local minimum (which is always the leftmost minimum in this case) has a value slightly larger than $0.9$ that increases slowly with $n$.
An example of the polynomial for $n=9$ is given in \cref{hermite_ctrl_phase}.

\begin{figure}[t]
    \centering
    \includegraphics[width=\columnwidth]{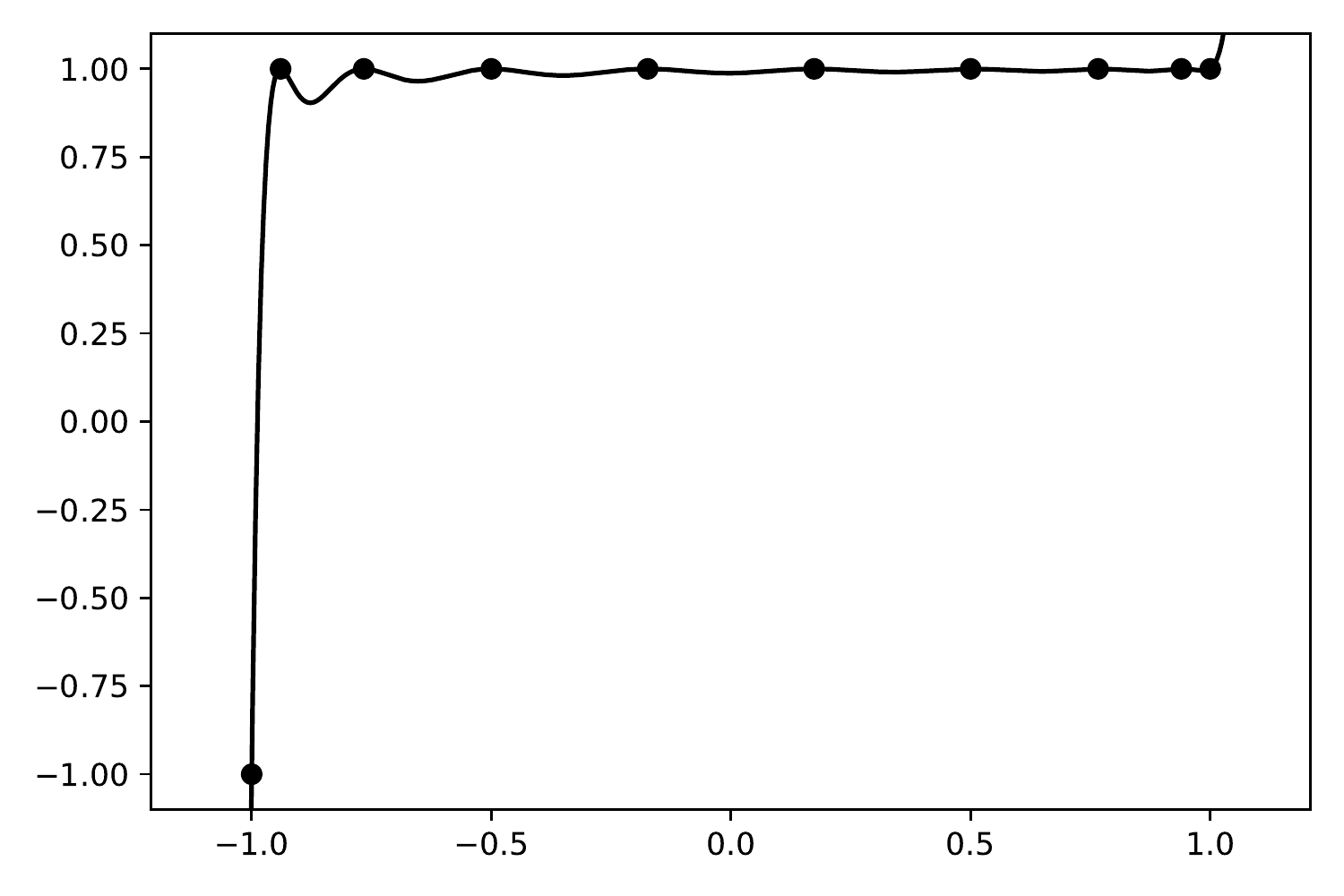}
    \caption{Example of the Hermite interpolating polynomial for an $n$-qubit controlled phase gate for $n=9$. As in \cref{encoding_operations}, $x$-coordinates of points are eigenvalues of $\cos\left(\frac{\pi}{2}\left(\mathds{1}-\frac{1}{n}\sum_{j}Z_j\right)\right)$ while $y$-coordinates are the phases controlled on them, so the leftmost point corresponds to all qubits being $1$, in which case the $-1$ phase is applied.}
    \label{hermite_ctrl_phase}
\end{figure}

Conjugating this multiply-controlled phase by single-qubit Hadamards acting on any one qubit creates a multiply-controlled NOT acting on that qubit and controlled on the others.
This provides an alternative construction of a multiply-controlled NOT to the usual method using $O(n)$ ancilla qubits and a cascade of Toffoli gates~\cite{nielsen01}.
Our method instead uses only one ancilla qubit and is compiled directly into singly-controlled gates (instead of Toffoli gates), at the expense of requiring $O(n^2)$ gates rather than $O(n)$.
Hence, this construction represents another space-versus-time tradeoff in the same vein as the main topic of this paper, only with a different application.

~
\begin{acknowledgements}
The authors thank Daniel Miller, Mark Steudtner, Mario Motta, and Peter Love for helpful conversations. W.~M.~K. acknowledges support from the National Science Foundation, Grant No. DGE-1842474.
\end{acknowledgements}

\bibliographystyle{ieeetr}
\bibliography{references}

\appendix

\section{Proofs}
\label[appendix]{proofs}

Lemmas whose numbers are preceded by `A' are referenced but not stated in the main text, while all other results are stated in the main text.
The results appear in the order in which they are referenced or stated in the main text.

\begin{lemma}
\label{overlap_lemma}
    Let $\mathcal{B}$ be a set of length-$Q$, Hamming-weight-$L$ bitstrings whose overlaps are upper-bounded by $D$, i.e., $\alpha\cdot\beta\le D$ for any distinct $\alpha,\beta\in\mathcal{B}$.
    Let $\gamma$ be a bitwise sum of $n\le G$ elements in $\mathcal{B}$, for any $G$ such that
    \begin{equation}
        L>2DG.
    \end{equation}
    Then both of the following are true:
    \begin{enumerate}
        \item $\alpha$ appears in the sum that defines $\gamma$ if and only if
        \begin{equation}
            \alpha\cdot\gamma=\sum_{i=0}^{Q-1}\alpha_i\gamma_i>\frac{L}{2}.
        \end{equation}
        \item all sums of up to $G$ elements in $\mathcal{B}$ are unique.
    \end{enumerate}
\end{lemma}
\noindent
Note: in \cref{overlap_lemma}, $\mathcal{B}$ is the set of elementary codewords, i.e., the image of $\mathcal{E}$ acting on the set of Hamming-weight-one bitstrings (elementary bitstrings); see~\cref{encoding_states} for definitions.
\begin{proof}
Let $\alpha,\beta^{(1)},...,\beta^{(n)}$ be distinct elements of $\mathcal{B}$ for $n\le G$ .
By assumption, $\alpha\cdot\beta^{(j)}\le D$ for all $j$.
Hence if
\begin{equation}
    \gamma=\beta^{(1)}\oplus\cdots\oplus\beta^{(n)}
\end{equation}
where $\oplus$ denotes bitwise sum, then
\begin{equation}
    \alpha\cdot \gamma=\alpha\cdot(\beta^{(1)}\oplus\cdots\oplus\beta^{(n)})\le\sum_{j=1}^n\alpha\cdot\beta^{(j)}\le Dn\le DG
\end{equation}
by the triangle inequality.
In other words, when $\alpha$ is not included in the sum that defines $\gamma$,
\begin{equation}
    \alpha\cdot\gamma\le DG<L/2.
\end{equation}

However, if
\begin{equation}
    \gamma=\alpha\oplus\beta^{(1)}\oplus\cdots\oplus \beta^{(n-1)},
\end{equation}
then
\begin{equation}
\begin{split}
    \alpha\cdot\gamma&=\alpha\cdot(\alpha\oplus\beta^{(1)}\oplus\cdots\oplus \beta^{(n-1)})\\
    &\ge \alpha\cdot\alpha-\sum_{j=1}^{n-1}\alpha\cdot\beta^{(j)}\\
    &\ge L-(n-1)D\\
    &\ge L-(G-1)D
\end{split}
\end{equation}
by the reverse triangle inequality.
In other words, when $\alpha$ is included in the sum that defines $\gamma$,
\begin{equation}
    \alpha\cdot\gamma\ge L-(G-1)D>L/2.
\end{equation}
This completes the proof of claim 1 in the lemma statement.

Claim 2 in the lemma statement follows from this because claim 1 provides a method for determining whether $\alpha$ is in a sum of up to $G$ elements of $\mathcal{B}$, for each $\alpha\in\mathcal{B}$.
Hence, given the sum, we can identify which elements of $\mathcal{B}$ formed it, which would be impossible if not all such sums were unique.

\end{proof}

\noindent
\textbf{Theorem~\ref{performance_theorem}.}
\emph{
    In the $M\gg F$ limit, our code satisfies
    \begin{equation}
        Q=O\left(F^2\log^4 M\right),
    \end{equation}
    with $D$ satisfying $D=O(\log M)$.
}
\begin{proof}

For given $D$, $M$, and $L=2DF\lceil\log_2(M+1)\rceil+1$, we encode the $M$ modes in $L'L$ qubits, where $L'$ is the least prime such that
\begin{equation}
    (L')^{D+1}\ge M\quad\text{and}\quad L'\ge L.
\end{equation}
Hence, for fixed $D$ the total number of qubits required in the large-$M$ limit is \eqref{qubits_exact} in the main text, which we reproduce here for convenience:
\begin{equation}
\label{qubits_exact_app}
    Q=M^\frac{1}{D+1}(2DF\lceil\log_2(M+1)\rceil+1)+O(F\log^2M),
\end{equation}
for which the corresponding $L'$ is
\begin{equation}
\label{Lprime_max}
    L'=M^\frac{1}{D+1}+O\left(\frac{\log M}{D+1}\right),
\end{equation}
since $L'$ is the least prime greater than or equal to $M^\frac{1}{D+1}$, and by the prime number theorem, the least prime greater than $M^\frac{1}{D+1}$ exceeds it by $O\left(\log\left(M^\frac{1}{D+1}\right)\right)$ on average.

However, in order to find the optimal value of $D$, we want to allow $D$ to be a function of $M$, in which case the constraint $L'\ge L$ means that we should modify \eqref{Lprime_max} to
\begin{equation}
\label{Lprime_max_mod}
    L'\approx\max\left\{M^\frac{1}{D+1},2DF\lceil\log_2M\rceil+1\right\},
\end{equation}
and \eqref{qubits_exact_app} correspondingly becomes
\begin{equation}
\label{qubits_exact_mod}
\begin{split}
    Q\approx&\max\left\{M^\frac{1}{D+1},2DF\lceil\log_2M\rceil+1\right\}\\
    &~\cdot(2DF\lceil\log_2M\rceil),
\end{split}
\end{equation}
where $\approx$ indicates that subleading terms are suppressed.
Therefore, in the large-$M$ limit the best choice of $D$ is whatever value minimizes \eqref{qubits_exact_mod}.
Equivalently, we want to minimize
\begin{equation}
\label{cost_function_D}
    \max\left\{DM^\frac{1}{D+1},2D^2F\lceil\log_2 M\rceil+1\right\}
\end{equation}
over $D$.

First, take a derivative of $DM^\frac{1}{D+1}$ (the first argument of the max above) with respect to $D$, set equal to zero, and solve, which results in
\begin{equation}
\label{minimizing_D}
    D=D^*\equiv\frac{\log(M)-2\pm\sqrt{\log^2(M)-4\log(M)}}{2}.
\end{equation}
Note that when we use $\log$ without an explicit base, we mean natural logarithm.
Both solutions are positive, and one can verify that the smaller value of $D^*$ corresponds to a local maximum of $DM^\frac{1}{D+1}$ and the larger corresponds to a local minimum.
Hence, $DM^\frac{1}{D+1}$ decreases monotonically between the two values of $D^*$ given in \eqref{minimizing_D}.

However, it turns out that the minimum of $DM^\frac{1}{D+1}$ (at the larger value of $D^*$) is smaller than the second argument of the max in \eqref{cost_function_D} evaluated at the same point.
To see this, note that for $D^{**}$ defined by
\begin{equation}
    D^{**}=\frac{1}{2}\log(M)-1,
\end{equation}
$D^{**}$ is smaller than the larger value of $D^*$ in \eqref{minimizing_D}.
Evaluating each argument of the max in \eqref{cost_function_D} at $D^{**}$ yields
\begin{equation}
\begin{split}
    &D^{**}M^\frac{1}{D^{**}+1}=D^{**}e^2=\Theta(\log M),\\
    &2(D^{**})^2F\lceil\log_2 M\rceil+1=\Theta(F\log^3 M),
\end{split}
\end{equation}
so since $2D^2F\lceil\log_2 M\rceil+1$ grows with $D$ and $DM^\frac{1}{D+1}$ decreases between $D^{**}$ and its actual minimum at $D^*>D^{**}$, the two arguments of the max in \eqref{cost_function_D} cross between the two values of $D^*$.
Therefore, the minimum of \eqref{cost_function_D} is the point where the two arguments of the max in \eqref{cost_function_D} are equal, which gives
\begin{equation}
\label{min_D_condition}
    M^\frac{1}{D+1}=2DF\lceil\log_2 M\rceil+1.
\end{equation}

At this point, \eqref{qubits_exact_app} becomes
\begin{equation}
\label{qubits_penultimate}
    Q=\left(2DF\lceil\log_2 M\rceil+1\right)^2+O(F\log^2 M).
\end{equation}
We can evaluate $D$ by taking the log of \eqref{min_D_condition} and rearranging to obtain
\begin{equation}
    D=\frac{\log M}{\log\left(2DF\lceil\log_2 M\rceil+1\right)}-1.
\end{equation}
We could apply this formula recursively to obtain arbitrarily good approximations of $D$, but for the purpose of this proof, we instead simply observe that it is upper bounded by
\begin{equation}
    D\le\log M,
\end{equation}
which when inserted in \eqref{qubits_penultimate} yields our final expression,
\begin{equation}
\begin{split}
    Q&=\left(2F\lceil\log_2 M\rceil\log M+1\right)^2+O(F\log^2 M)\\
    &=O\left(F^2\log^4M\right).
\end{split}
\end{equation}

\end{proof}

\noindent
\textbf{Lemma~\ref{interpolation_lemma}.}
\emph{
    For odd $L$, the polynomial $A$ obtained by Hermite interpolation of the points \eqref{hermite_points}, together with the constraints that its first derivative be zero at all non-edge points, satisfies:
    \begin{equation}
        \label{poly_constraints_proof_app}
        \begin{split}
            &A(\lambda)\le1\quad\forall\lambda\in[0,1],\\
            &A(\lambda)\ge-1\quad\forall\lambda\in[-1,0],\\
            &|A(\lambda)|\ge1\quad\forall\lambda\notin(-1,1).
        \end{split}
    \end{equation}
    The $x$-coordinates of the non-edge points are
    \begin{equation}
        \{\cos\left(\frac{m\pi}{L}\right)~|~m=1,2,...,L-2,L-1\}.
    \end{equation}
    Note: an example of the Hermite interpolating polynomial for $L=9$ is given in \cref{hermite_example}
}

\begin{proof}

By definition, $A$ is the least-degree polynomial that satisfies the given constraints, which implies that it has degree $2L-1$ because there are $2L$ constraints.
Hence, its derivative $A'$ has degree $2L-2$, so $A$ has at most $2L-2$ local extrema.
By construction, one extremum is located at each non-edge point, of which there are $L-1$.
The remaining $L-1$ extrema must therefore be located between all pairs of adjacent points for which the values of the polynomial are the same (i.e., all adjacent pairs except for the middle pair), since there are $L-1$ such pairs.
This follows because for any pair of adjacent points, the polynomial has zero slope at at least one of the points, and it cannot be a straight line between the points, so in order to pass through the other point it must have an extremum between the points.
For a visual aid, see \cref{hermite_example}.
Therefore, all extrema of the polynomial are either located at interpolated points, or between interpolated points with the same values.

Hence, there is no extremum between the middle pair of points
\begin{equation}
    (\cos\left(\frac{\pi}{L}\left\lceil\frac{L}{2}\right\rceil\right),-1)\quad\text{and}\quad(\cos\left(\frac{\pi}{L}\left\lfloor\frac{L}{2}\right\rfloor\right),1),
\end{equation}
where the values switch from negative to positive.
This implies that the slope of the polynomial must be positive between these points, i.e., it must approach the point $p^{+}=(\cos\left(\frac{\pi}{L}\left\lfloor\frac{L}{2}\right\rfloor\right),1)$ from below, and the point $p^{-}=(\cos\left(\frac{\pi}{L}\left\lceil\frac{L}{2}\right\rceil\right),-1)$ from above.

We established above that the derivative $A'$ has its $2L-2$ roots at each interpolation point as well as between all interpolation points with the same values.
This means that $A'$ must have extrema between each of its roots, which accounts for $2L-3$ extrema.
However, since these correspond to roots of the second-derivative $A''$, which has degree $2L-3$, they must account for all of its roots, i.e., $A$ has points with zero curvature only between its extrema.
Together with the fact that $p^{+}$ is an extremum of $A$ and $A$ approaches $p^{+}$ from below, this implies that $A$ must have negative curvature at $p^{+}$.
Hence, $A$'s next extremum to the right of $p^{+}$ must be below $p^{+}$, so $A$ must approach the next interpolation point to the right of $p^{+}$ from below, and so forth.
This means that the first line in \eqref{poly_constraints_proof_app} holds with equality only at the interpolation points.
A similar argument implies that the second line in \eqref{poly_constraints_proof_app} holds with equality only at the interpolation points.

Also, the above argument implies that $A$ must approach the rightmost point $(1,1)$ from below.
We also established that $A$ cannot have an extremum either at or to the right of $(1,1)$, which means that $A(\lambda)$ must continue to grow for $\lambda\ge1$, i.e., $A(\lambda)\ge1$ for $\lambda\ge1$.
Similarly, we find that $A(\lambda)\le-1$ for $\lambda\le-1$.
This proves the third line in \eqref{poly_constraints_proof_app}.

\end{proof}

\noindent
\textbf{Lemma~\ref{phased_iterate_implementation_lemma}.}
\emph{
    For
    \begin{equation}
    \label{R_def}
        R_\phi\equiv
        \begin{pmatrix}
            1&0\\
            0&e^{i\phi}
        \end{pmatrix}_a
        \otimes\mathds{1},
    \end{equation}
    $W_\phi$ is given by the following sequence of operations:
    \begin{equation}
    \label{Wphi_decomp}
    \begin{split}
        W_{\phi}=&R_\phi(H\otimes\mathds{1})\\
        &\cdot R_\pi\left(\prod_{j\in S_i}\text{ctrl-}e^{-\frac{i\pi Z_j}{L}}\right)\left(\prod_{j\in S_i}e^{-\frac{i\pi Z_j}{2L}}\right)\\
        &\cdot(H\otimes\mathds{1})R_\phi^\dagger,
    \end{split}
    \end{equation}
    where $H$ is the Hadamard gate, $Z_j$ is a single-qubit Pauli-$Z$ acting on code qubit $j$, and the controls are on the ancilla qubit.
}

\begin{proof}
The space that $W_\phi$ acts upon is the tensor product of a single ancilla $|\cdot\rangle_a$ and the computational space (with computational basis states $|q\rangle$) that we want $Z^{(BK)}_i$ to act upon.
We implement $W_\phi$ as follows, for $\mathcal{H}_i$ and $\mathcal{G}_i$ defined by \eqref{Hm_def} and \eqref{Gm_def}, respectively:
\begin{equation}
\begin{split}
    &\begin{pmatrix}
        |0\rangle_a\\
        |1\rangle_a
    \end{pmatrix}
    \otimes|q\rangle
    \xrightarrow{H\otimes\mathds{1}}
    \frac{1}{\sqrt{2}}(|0\rangle\pm|1\rangle)\otimes|q\rangle\\
    &\xrightarrow{e^{-i\mathcal{G}_i}}
    \frac{1}{\sqrt{2}}\Big(|0\rangle\otimes e^{-i\mathcal{G}_i}|q\rangle\pm|1\rangle\otimes e^{-i\mathcal{G}_i}|q\rangle\Big)\\
    &\xrightarrow{\text{ctrl-$e^{2i\mathcal{G}_i}$}}
    \frac{1}{\sqrt{2}}\Big(|0\rangle\otimes e^{-i\mathcal{G}_i}|q\rangle\pm|1\rangle\otimes e^{i\mathcal{G}_i}|q\rangle\Big)\\
    &\xrightarrow{H\otimes\mathds{1}}
    \Big(\frac{e^{-i\mathcal{G}_i}\pm e^{i\mathcal{G}_i}}{2}|0\rangle+\frac{e^{-i\mathcal{G}_i}\mp e^{i\mathcal{G}_i}}{2}|1\rangle\Big)\otimes|q\rangle\\
    &=\begin{pmatrix}
        \cos\mathcal{G}_i&-i\sin\mathcal{G}_i\\
        -i\sin\mathcal{G}_i&\cos\mathcal{G}_i
    \end{pmatrix}
    \begin{pmatrix}
        |0\rangle_a\\
        |1\rangle_a
    \end{pmatrix}
    \otimes|q\rangle\\
    &=\begin{pmatrix}
        \mathcal{H}_i&-i\sqrt{1-\mathcal{H}_i^2}\\
        -i\sqrt{1-\mathcal{H}_i^2}&\mathcal{H}_i
    \end{pmatrix}
    \begin{pmatrix}
        |0\rangle_a\\
        |1\rangle_a
    \end{pmatrix}
    \otimes|q\rangle,
\end{split}
\end{equation}
where the upper (lower) entries in the vector expressions correspond to the upper (lower) values of the $\pm$s and $\mp$s, and the $e^{-i\mathcal{G}_i}$ in the second-to-last line just yields an overall phase.
Hence
\begin{equation}
    W_{\phi=0}=(H\otimes\mathds{1})(\text{ctrl-$e^{2i\mathcal{G}_i}$})e^{-i\mathcal{G}_i}(H\otimes\mathds{1}).
\end{equation}
To obtain $W_\phi$ for $\phi\neq0$, we conjugate this by the phases on the ancilla qubit~\cite{low2019qubitization}, denoted by $R_\phi$ as defined in \eqref{R_def}:
\begin{equation}
\begin{split}
    W_{\phi}&=R_\phi W_0R_\phi^\dagger\\
    &=R_\phi(H\otimes\mathds{1})(\text{ctrl-$e^{2i\mathcal{G}_i}$})e^{-i\mathcal{G}_i}(H\otimes\mathds{1})R_\phi^\dagger.
\end{split}
\end{equation}
Finally, since $\mathcal{G}_i$ is defined by \eqref{Gm_def},
\begin{equation}
\begin{split}
    e^{2i\mathcal{G}_i}&=\exp\left(i\pi\left(\mathds{1}-\frac{1}{L}\sum_{j\in S_i}Z_j\right)\right)\\
    &=-\exp\left(-\frac{i\pi}{L}\sum_{j\in S_i}Z_j\right)\\
    &=-\prod_{j\in S_i}e^{-i\pi Z_j/L},
\end{split}
\end{equation}
so
\begin{equation}
    \text{ctrl-}e^{2i\mathcal{G}_i}=R_\pi\prod_{j\in S_i}\text{ctrl-}e^{-\frac{i\pi Z_j}{L}},
\end{equation}
i.e., $\text{ctrl-}e^{2i\mathcal{G}_i}$ decomposes into a product of single-qubit phases controlled by the ancilla qubit.
Similarly,
\begin{equation}
    e^{-i\mathcal{G}_i}=i\prod_{j\in S_i}e^{-\frac{i\pi Z_j}{2L}},
\end{equation}
with the factor of $i$ on the right-hand side an irrelevant overall phase.
Thus our final decomposition of $W_\phi$ is \eqref{Wphi_decomp}.
\end{proof}

\noindent
\textbf{Corollary~\ref{ansatz_implementation_lemma}.}
\emph{
    We can implement an encoded \emph{hop gate}, which is universal for real-valued wavefunctions with fixed fermion number~\cite{eddins2021entanglementforging}, or the encoding of a unitary generated by any term in the Hamiltonian, using one ancilla qubit and a number of doubly-controlled (three-qubit) gates given by \eqref{term_cost}.
}
\begin{proof}

We first prove the second claim, the construction of an encoded unitary generated by any term in the Hamiltonian, since the proof is simpler.
Let $T$ denote the term in the Hamiltonian that we wish to use to generate a unitary.
In our encoding, $T$ is both unitary and Hermitian, and we can implement it as a unitary with cost given by \eqref{term_cost}.

Let $|x\rangle$ be an arbitrary eigenvector of $T$.
If we can implement a unitary that has the desired action on any arbitrary such $|x\rangle$, it must be exactly equal to the desired unitary, since $T$ is Hermitian and thus it possesses an eigenbasis spanning the Hilbert space.
To implement the unitary $e^{i\theta T}$, introduce a single ancilla qubit $|\cdot\rangle_b$ and perform the following operations:
\begin{equation}
\label{big_step_one}
\begin{split}
    |0\rangle_b\otimes|x\rangle\xrightarrow{H\otimes\mathds{1}}~&\frac{1}{\sqrt{2}}(|0\rangle_b+|1\rangle_b)\otimes|x\rangle\\
    \xrightarrow{\text{ctrl-$T$}}~&\frac{1}{\sqrt{2}}\Big(|0\rangle_b\otimes|x\rangle+|1\rangle_b\otimes(-1)^\tau|x\rangle\Big)\\
    &=\frac{1}{\sqrt{2}}\Big(|0\rangle_b+(-1)^\tau|1\rangle_b\Big)\otimes|x\rangle,
\end{split}
\end{equation}
where $(-1)^\tau$ for $\tau\in\{0,1\}$ is the eigenvalue of $T$ for the state $|x\rangle$.
Continuing from where we left off...
\begin{equation}
\label{big_step_two}
\begin{split}
    &\vdots\\
    =&\frac{1}{\sqrt{2}}\Big(|0\rangle_b+(-1)^\tau|1\rangle_b\Big)\otimes|x\rangle\\
    \xrightarrow{H\otimes\mathds{1}}~&
    \begin{cases}
        |0\rangle_b\otimes|x\rangle\quad\text{if $\tau=0$},\\
        |1\rangle_b\otimes|x\rangle\quad\text{if $\tau=1$}
    \end{cases}\\
    \xrightarrow{R_{-2\theta}}~&
    \begin{cases}
        |0\rangle_b\otimes|x\rangle\quad\text{if $\tau=0$},\\
        e^{2i\theta}|1\rangle_b\otimes|x\rangle\quad\text{if $\tau=1$}
    \end{cases}\\
    \xrightarrow{H\otimes\mathds{1}}~&\frac{e^{-2i\theta\tau}}{\sqrt{2}}\Big(|0\rangle_b+(-1)^\tau|1\rangle_b\Big)\otimes|x\rangle\\
    \xrightarrow{\text{ctrl-$T$}}~&\frac{e^{-2i\theta\tau}}{\sqrt{2}}(|0\rangle_b+|1\rangle_b)\otimes|x\rangle\\
    \xrightarrow{H\otimes\mathds{1}}~&e^{-2i\theta\tau}|0\rangle_b\otimes|x\rangle\\
    &=e^{-i\theta}|0\rangle_b\otimes e^{i\theta T}|x\rangle.
\end{split}
\end{equation}
Thus, we have implemented the desired operation up to an overall phase.
The dominant costs in this construction were the two controlled applications of $T$, each of which has cost given by \eqref{term_cost} but in doubly-controlled gates rather than singly-controlled gates, since \eqref{term_cost} is the cost of a non-controlled application of $T$.

The \emph{hop gate} is a gate that acts on two fermionic modes as
\begin{equation}
    h(\varphi)=
    \begin{pmatrix}
        1&0&0&0\\
        0&\cos\varphi&-\sin\varphi&0\\
        0&\sin\varphi&\cos\varphi&0\\
        0&0&0&-1
    \end{pmatrix},
\end{equation}
which we can decompose as
\begin{equation}
\label{hop_decomp}
    h(\varphi)=
    \begin{pmatrix}
        1&0&0&0\\
        0&\cos\varphi&-\sin\varphi&0\\
        0&\sin\varphi&\cos\varphi&0\\
        0&0&0&1
    \end{pmatrix}
    \begin{pmatrix}
        1&0&0&0\\
        0&1&0&0\\
        0&0&1&0\\
        0&0&0&-1
    \end{pmatrix}.
\end{equation}
In other words, we can think of the hop gate as first applying a controlled phase, and then rotating occupation between the two modes.

It will be useful to first decompose the hop gate into fermionic Pauli operators $X^{(f)}$, $Y^{(f)}$, and $Z^{(f)}$, meaning Pauli operators applied directly as unitaries acting on fermionic modes.
In terms of these, the controlled phase in \eqref{hop_decomp} can be implemented as follows.
Introduce an ancilla qubit initially in state $|0\rangle_b$, distinct from the ancilla $|\cdot\rangle_a$ used for quantum signal processing as described in \cref{encoding_states} in the main text.
Let $|z\rangle$ denote an encoded occupation number state, i.e., an eigenstate of $Z^{(f)}$ acting on every fermionic mode.
Such states are a basis for the fermionic Hilbert space, so if we can implement the controlled phase for an arbitrary $|z\rangle$, then the same implementation will apply it to an arbitrary fermionic state.

If the two modes to which the controlled phase is to be applied are $i$ and $j$, implement
\begin{equation}
\begin{split}
    |0\rangle_b\otimes|z\rangle\xrightarrow{H\otimes\mathds{1}}~&\frac{1}{\sqrt{2}}(|0\rangle_b+|1\rangle_b)\otimes|z\rangle\\
    \xrightarrow{\text{ctrl-}Z^{(f)}_i}~&\frac{1}{\sqrt{2}}\Big(|0\rangle_b\otimes|z\rangle+|1\rangle_b\otimes Z^{(f)}_i|z\rangle\Big)\\
    \xrightarrow{H\otimes\mathds{1}}~&
    \begin{cases}
        |0\rangle_b\otimes|z\rangle\quad\text{if $Z^{(f)}_i|z\rangle=|z\rangle$},\\
        |1\rangle_b\otimes|z\rangle\quad\text{if $Z^{(f)}_i|z\rangle=-|z\rangle$}
    \end{cases}\\
    \xrightarrow{\text{ctrl-}Z^{(f)}_j}~&
    \begin{cases}
        |0\rangle_b\otimes\zeta|z\rangle\quad\text{if $Z^{(f)}_i|z\rangle=|z\rangle$},\\
        |1\rangle_b\otimes\zeta|z\rangle\quad\text{if $Z^{(f)}_i|z\rangle=-|z\rangle$}
    \end{cases},
\end{split}
\end{equation}
where $\zeta=-1$ if $Z^{(f)}_i|z\rangle=-|z\rangle$ and $Z^{(f)}_j|z\rangle=-|z\rangle$, and $\zeta=1$ otherwise, i.e., $\zeta$ is the desired phase due to the controlled phase operation.
All that remains is to uncompute the ancilla, which we can do by reversing the first three operations above.
Thus we implement the controlled phase via three controlled applications of $Z^{(f)}$ operators, together with four single-qubit gates.

Next, we implement the rotation in \eqref{hop_decomp}, which may be rewritten as
\begin{equation}
\label{phased_swap}
\begin{split}
    &\begin{pmatrix}
        1&0&0&0\\
        0&\cos\varphi&-\sin\varphi&0\\
        0&\sin\varphi&\cos\varphi&0\\
        0&0&0&1
    \end{pmatrix}
    =e^{i\varphi(X^{(f)}Y^{(f)}-Y^{(f)}X^{(f)})/2}\\
    &=e^{\frac{i\varphi X^{(f)}Y^{(f)}}{2}}e^{-\frac{i\varphi Y^{(f)}X^{(f)}}{2}}\\
    &=e^{\frac{i\pi I^{(f)}Z^{(f)}}{4}}e^{\frac{i\varphi X^{(f)}X^{(f)}}{2}}e^{-\frac{i\pi I^{(f)}Z^{(f)}}{4}}\\
    &\quad\cdot e^{\frac{i\pi Z^{(f)}I^{(f)}}{4}}e^{\frac{i\varphi X^{(f)}X^{(f)}}{2}}e^{-\frac{i\pi Z^{(f)}I^{(f)}}{4}},
\end{split}
\end{equation}
where tensor product symbols are suppressed, e.g., $X^{(f)}X^{(f)}=X^{(f)}\otimes X^{(f)}$.
Note that this sequence of operations can require intermediate states containing at most two more fermions than the original state, although the final state must have the same fermion number.
As discussed in the main text, to encode the $F$-fermion Hamiltonian we in fact implement the encoding for all states of up to $F+4$ fermions.
Therefore, this also implies that the sequence of operations \eqref{phased_swap} will have the desired action.

The operations in \eqref{phased_swap} are rotations generated by the encoded operators $Z^{(f)}_i$ and $X^{(f)}_iX^{(f)}_j$.
These rotations can be implemented using the same method as for the rotation generated by a term $T$ in the Hamiltonian, above, with the following changes:
\begin{enumerate}
    \item replace the initial state $|x\rangle$ with an eigenstate of the operator generating the rotation;
    
    \item replace the rotation angle $\theta$ in \eqref{big_step_two} with $\pi/2$ or $\pm\pi/4$, depending on which rotation in \eqref{phased_swap} is desired;
    
    \item replace the controlled-$T$ operations in \eqref{big_step_one} and \eqref{big_step_two} with controlled applications of the desired generator.
\end{enumerate}

Hence, we can implement the entire hop gate using $O(1)$ controlled applications of $X^{(f)}_iX^{(f)}_j$ and $Z^{(f)}_i$, as well as $O(1)$ single-qubit gates.
Under the Bravyi-Kitaev mapping, each $X^{(f)}_i$ becomes a product of $O(\log M)$ $X^{(BK)}$ operators, so $X^{(f)}_iX^{(f)}_j$ also becomes a product of $O(\log M)$ $X^{(BK)}$ operators.
Similarly, each $Z^{(f)}_i$ becomes a product of $O(\log M)$ $Z^{(BK)}$ operators.
Each $Z^{(BK)}$ has cost given by \eqref{parity_cost}, so since we require $O(\log M)$ of them the overall cost becomes \eqref{term_cost} in doubly-controlled gates for the same reason as in the construction of the rotation generated by a term in the Hamiltonian.

\end{proof}

\begin{lemma}
\label{fixed_lattice_lemma}
    On a linear qubit architecture, where two-qubit operations can only be performed on adjacent qubits in the line, the number of qubit swaps required to implement $Z^{(BK)}_i$ as defined by \eqref{Z_action} is $O(QL)$, where $Q$ is the number of qubits and $L=2DG+1$ (Eq.~\eqref{L_min} in the main text).
\end{lemma}
\begin{proof}

As shown in the proof of \cref{parity_theorem}, $Z^{(BK)}_i$ is implemented via $2L-1$ applications of the quantum signal processing iterate $W_\phi$.
$W_\phi$ is implement as in \eqref{Wphi_decomp}, so the only two-qubit operations in the implementation of $Z^{(BK)}_i$ are the controlled phases in \eqref{Wphi_decomp}:
\begin{equation}
\label{controlled_phases}
    \prod_{j\in S_i}\text{ctrl-}e^{-i\pi Z_j/L}.
\end{equation}
These $L$ controlled phases are all controlled on the same qubit, the quantum signal processing ancilla ($|\cdot\rangle_a$).

Hence, on a line of qubits, we can successively swap this control qubit along the line so that it is adjacent to each qubit it needs to control, which are the qubits in $S_i$ as in \eqref{controlled_phases}.
Since the controlled phases in \eqref{controlled_phases} all commute, the order in which they are applied is irrelevant.
Therefore, given any initial location of the control qubit in the line, we can classically choose a path whose length is upper bounded by $3Q/2$ that brings the control qubit adjacent to each qubit in $S_i$.
The worst case is when the control qubit is initially in the center of the line, and $S_i$ contains the qubits at both ends of the line, in which case the shortest path is to first swap the control qubit to the end of the line it is closer to, and then swap it back along the whole line.
Since this path brings the control qubit adjacent to all other qubits (not just those in $S_i$), there can be no worse case.

Therefore, each implementation of the sequence of controlled phases and hence each implementation of $W_\phi$ requires at most $3Q/2$ swaps.
Since $Z^{(BK)}_i$ requires $2L-1$ applications of $W_\phi$, it requires
\begin{equation}
    \frac{3Q}{2}(2L-1)=O(QL)
\end{equation}
swaps.

\end{proof}

\section[appendix]{Threshold for outperforming Jordan-Wigner and Bravyi-Kitaev}
\label{threshold_app}

As discussed in the main text, the minimum number of modes for which our encoding is advantageous over Jordan-Wigner and Bravyi-Kitaev occurs when $D=1$.
In this case, $L=2G+1$ for $G=F\lceil\log_2(M+1)\rceil$, so for
\begin{equation}
    L'=\textsc{NextPrime}(2G+1)
\end{equation}
the least prime greater than $2G+1$, the number $(L')^2$ of modes we can encode is greater than the number $L'L$ of qubits.
This means that for
\begin{equation}
\begin{split}
    &LL'=(2G+1)\textsc{NextPrime}(2G+1)\\
    &<M\\
    &\le(L')^2=\big(\textsc{NextPrime}(2G+1)\big)^2
\end{split}
\end{equation}
our encoding is advantageous over Bravyi-Kitaev.

However, depending on the gaps between primes greater $\textsc{NextPrime}(2G+1)$, there may be one or more subsequent ranges of $M$ in which the encoding reduces to Bravyi-Kitaev.
Let
\begin{equation}
    \textsc{NextPrime}^k(2G+1)
\end{equation}
denote the $k$th prime greater than $2G+1$.
Then if for any $k=1,2,...$,
\begin{equation}
\label{jw_condition}
\begin{split}
    &\big(\textsc{NextPrime}^k(2G+1)\big)^2\\
    &<(2G+1)\textsc{NextPrime}^{k+1}(2G+1),
\end{split}
\end{equation}
our encoding will reduce to Bravyi-Kitaev for any values of $M$ contained in
\begin{equation}
\begin{split}
    \Big(\big(&\textsc{NextPrime}^k(2G+1)\big)^2,\\
    &(2G+1)\textsc{NextPrime}^{k+1}(2G+1)\Big],
\end{split}
\end{equation}
since $M$ is larger than
\begin{equation}
    \big(\textsc{NextPrime}^k(2G+1)\big)^2
\end{equation}
the number of modes that can be encoded in
\begin{equation}
    (2G+1)\textsc{NextPrime}^k(2G+1)
\end{equation}
qubits, but smaller than the number of qubits
\begin{equation}
    (2G+1)\textsc{NextPrime}^{k+1}(2G+1)
\end{equation}
required for the next code size.
However, since the gaps between primes are on average logarithmic in the sizes of the primes, for all but at most a few small values of $k$, \eqref{jw_condition} will not hold and thus the corresponding ranges will be empty, so our encoding will be advantageous.
For $L=2G+1$ up to $501$ (corresponding to at least $L^2=251001$ qubits), we directly checked the maximum values of $k$ for which \eqref{jw_condition} holds, and found that in this range $k$ did not exceed four.

\section[appendix]{Application of quantum signal processing construction of fermion operators to segment code of \cite{steudtner2018fermions,steudtner2019fermions}}
\label{application_to_segment_code}

The construction in \cref{encoding_states} in the main text allows us to implement $Z^{(BK)}_i$ as given by \eqref{Z_action}.
In other words, given some set of $L$ qubits for odd $L$, we can implement a $-1$ phase controlled on more than half of the qubits being in state $|1\rangle$, i.e., on the Hamming weight of a computational basis state of the qubits being greater than $L/2$.
This requires $O(L^2)$ one- and two-qubit gates, as in \eqref{parity_cost}.

This operation is exactly that required to implement the ``binary switch" used to implement the ``segment code'' of \cite{steudtner2018fermions,steudtner2019fermions}.
The remainder of the segment code is linear, so the corresponding encoded operations are Pauli operators.
Our $L$ corresponds to $\hat{n}$ in \cite{steudtner2018fermions,steudtner2019fermions}, and in their code $G=F\lceil\log_2(M+1)\rceil$ is replaced by $F$ (which is $K$ in their notation).
Since we require $L$ to be odd, we set
\begin{equation}
    L=\hat{n}=2K+1=2F+1
\end{equation}
instead of $\hat{n}=2K$ as in \cite{steudtner2018fermions,steudtner2019fermions} (i.e., we just use one extra qubit per segment).

Hence, one can implement encoded fermionic operators for the segment code using $O(L^2)=O(F^2)$ one- and two-qubit operations, and the encoding maps each segment of $L+1$ fermionic modes to $L$ qubits.
Therefore, the number of qubits required is
\begin{equation}
\begin{split}
    Q&=\left\lfloor\frac{M}{L+1}\right\rfloor L+\left(M-\left\lfloor\frac{M}{L+1}\right\rfloor(L+1)\right)\\
    &<\frac{M}{L+1}L+L=\left(1-\frac{1}{L+1}+\frac{L}{M}\right)M.
\end{split}
\end{equation}
For $M\gg F\gg1$, this is approximately
\begin{equation}
    Q\approx\left(1-\frac{1}{2F}\right)M=\left(1-\frac{1}{2K}\right)M,
\end{equation}
which is the value quoted from \cite{steudtner2018fermions,steudtner2019fermions}.
As noted in the main text, since this encoding begins to be advantageous over Jordan-Wigner as soon as $M\ge L+1=2F+2$, while our encoding does not become advantageous until $M=\Omega(F^2)$, we recommend using the segment code to bridge this gap in the small-$M$ regime.

\section[appendix]{Polynomials over finite fields}
\label{finite_fields_app}

Let $\mathbb{Z}_n$ denote the ring of integers modulo $n$, i.e.,
\begin{equation}
    \mathbb{Z}_n=\{0,1,2,...,n-1\}
\end{equation}
and addition and multiplication are carried out modulo $n$.
When $n=L'$ for prime $L'$, $\mathbb{Z}_{L'}$ is a field as well as a ring, which roughly means that it also possesses a division operation that satisfies the same properties as the usual division over real or rational numbers.
Furthermore, it is a field of characteristic $L'$, which means that $L'$ is the least number such that
\begin{equation}
    \underbrace{x+x+\cdots+x}_\text{$L'$ copies}=0\quad\forall x\in\mathbb{Z}_{L'},
\end{equation}
which implies that there is no $y\in\mathbb{Z}_{L'}$ such that $yx=0$ for all $x\in\mathbb{Z}_{L'}$.
See \cite{gallian2016contemporary} for a thorough introduction to rings and fields.
All arithmetic operations in this section are assumed to be modulo the order of the ring or field presently under consideration.

For any $D\in\mathbb{Z}_{L'}$, a degree-$D$ polynomial over the finite field $\mathbb{Z}_{L'}$ is a formal expression
\begin{equation}
\label{general_polynomial}
    c_0x^0+c_1x^1+\cdots+c_Dx^D,
\end{equation}
where the $c_i$ are coefficients in $\mathbb{Z}_{L'}$, and $x$ is the variable or \emph{indeterminate}.
For our purposes, we can think of a formal polynomial as equivalent to the list of its coefficients, which uniquely specify it.
A formal polynomial induces a function $f:\mathbb{Z}_{L'}\rightarrow\mathbb{Z}_{L'}$, called the \emph{induced polynomial function}, by replacing the variable $x$ with a value in $\mathbb{Z}_{L'}$ and evaluating the resulting expression modulo $L'$.
In the main text, we simply referred to these functions themselves as polynomials, for simplicity, but here we will explicitly refer to them as (induced) polynomial functions.

Over general finite fields, distinct formal polynomials can induce the same polynomial function.
However, over $\mathbb{Z}_{L'}$ (for prime $L'$) all distinct formal polynomials of degree less than $L'$ induce distinct polynomial functions.
This follows from the well-known fact that every function over a finite field is a polynomial function.
To see how our desired property follows, first note that there are $(L')^{L'}$ distinct functions over $\mathbb{Z}_{L'}$.
Next, using Fermat's Little Theorem, which states that $x^{L'}=x$ modulo $L'$ for any $x\in\mathbb{Z}_{L'}$, we can reduce any arbitrary polynomial function to a polynomial function of degree less than $L'$.
Note that we cannot reduce away $x^{L'-1}$ if its coefficient is nonzero, because Fermat's Little Theorem only implies $x^{L'-1}=1$ for nonzero $x\in\mathbb{Z}_{L'}$.

Hence, the polynomial function induced by any arbitrary formal polynomial is identical to the polynomial function induced by a formal polynomial of degree less than $L'$, so every function over $\mathbb{Z}_{L'}$ is a polynomial function induced by a formal polynomial of degree less than $L'$.
A formal polynomial of degree less than $L'$ over $\mathbb{Z}_{L'}$ is uniquely characterized by its $L'$ coefficients (of $x^0,x^1,...,x^{L'-1}$, allowing any of the coefficients to be zero), so there are $(L')^{L'}$ formal polynomials of degree less than $L'$ over $\mathbb{Z}_{L'}$.
Therefore, all of these must induce distinct polynomial functions, because if any pair of them induced the same polynomial function then there would not be enough of them to match all of the $(L')^{L'}$ general functions.

In the main text, we do not use all polynomial functions of degree less than the order of the field, but only those up to some fixed degree $D$.
However, this $D$ is always less than the order of the field, so all such polynomial functions are distinct, which justifies our claim in the main text that there are $(L')^{D+1}$ of them.

That a degree-$D$ polynomial function over $\mathbb{Z}_{L'}$ can have at most $D$ roots follows similarly to the argument over the real numbers.
Polynomials over finite fields admit polynomial long division, so a polynomial function $f$ whose roots form a multiset $R\subseteq\mathbb{Z}_{L'}$ (including multiple copies of roots with multiplicities greater than one) can be factored as
\begin{equation}
    f(x)=g(x)\prod_{r\in R}(x-r)
\end{equation}
where $g(x)$ is some other polynomial over $\mathbb{Z}_{L'}$.
Thus since the product over $R$ is itself a polynomial of degree $|R|$, the number of roots, the degree of $f$ must be at least the number of roots.
These are the main facts about polynomials over finite fields used in the main text.

\section[appendix]{Hermite interpolation}
\label{hermite_interpolation_app}

Hermite interpolation is a method for finding the least-degree polynomial (over the real numbers) that satisfies a certain set of constraints.
A special case, Newton interpolation, applies when the constraints are simply a set of specified points, i.e., function values at particular inputs.
In this case, when $n$ points are specified, the least-degree polynomial that passes through the points has degree $n-1$: for example, any single point defines the constant polynomial whose value is the value at the point, any pair of points defines a line, and so forth.

Hermite interpolation generalizes this to cases where up to $m$th derivatives are also specified at each point.
Different numbers of derivatives can be specified for different points.
In this case, each derivative and each point is a constraint, and if there are $n$ constraints in total then again the least-degree polynomial that satisfies the constraints has degree $n-1$.
For a thorough review of Hermite interpolation, see \cite{burden2015numerical}.

In this section, we will instead illustrate Hermite interpolation by showing how to implement it for the specific example in the main text.
In that example, for some odd $L$, the points are \eqref{hermite_points}, which we reproduce here for convenience
\begin{equation}
\begin{split}
    &\{\left(\cos\left(\frac{m\pi}{L}\right),1\right)~|~m=0,1,2,...,\left\lfloor\frac{L}{2}\right\rfloor\}\\
    &\cup\{\left(\cos\left(\frac{m\pi}{L}\right),-1\right)~|~m=\left\lfloor\frac{L}{2}\right\rfloor+1,...,L\}.
\end{split}
\end{equation}
The derivative constraints are that the first derivatives be zero at all points except for the first and last ($m=0$ and $m=L$).

Let us translate these constraints into a more general language: the value and derivatives at each point will be written as a list of numbers $(x_i,f_i,f'_i,f''_i,...)$, which stands for the constraints
\begin{equation}
    f(x_i)=f_i,\quad f'(x_i)=f'_i,\quad f''(x_i)=f''_i,...
\end{equation}
where $f$ is the polynomial we are trying to construct.
In this notation, the constraints we stated above may be rewritten
\begin{equation}
\label{constraints}
\begin{split}
    &(x_0,f_0)=(1,1),\\
    &(x_i,f_i,f'_i)=(\cos\left(\frac{i\pi}{L}\right),1,0)\quad\text{for $i=1,2,...,\left\lfloor\frac{L}{2}\right\rfloor$},\\
    &(x_i,f_i,f'_i)=(\cos\left(\frac{i\pi}{L}\right),-1,0)\quad\text{for $i=\left\lceil\frac{L}{2}\right\rceil,...,L-1$},\\
    &(x_L,f_L)=(-1,-1).
\end{split}
\end{equation}

To implement Hermite interpolation, we construct a second list $\{z_i\}$.
$\{z_i\}$ should be a list of the $x_i$s, in order, but with each $x_i$ duplicated if its first derivative is specified: in other words,
\begin{equation}
\begin{split}
    &z_0=x_0,\\
    &z_1=z_2=x_1,\\
    &z_3=z_4=x_2,\\
    &\qquad\vdots\\
    &z_{2i-1}=z_{2i}=x_i,\\
    &\qquad\vdots\\
    &z_{2L-3}=z_{2L-2}=x_{L-1},\\
    &z_{2L-1}=x_L.
\end{split}
\end{equation}
Then the expression for the Hermite interpolating polynomial is
\begin{equation}
\label{interpolating_polynomial_expression}
    f(x)=\sum_{i=0}^{2L-1}f[z_i,z_{i-1},...,z_1,z_0]\prod_{j=0}^{i-1}(x-z_j),
\end{equation}
where the product $\prod_{j=0}^{-1}(x-z_j)\equiv1$, and $f[z_i,z_{i-1},...,z_1,z_0]$ is the \emph{divided difference} of $f$, defined below.
We can see that is a degree-$(2L-1)$ polynomial, as we expect, since there are $L+1$ value constraints and $L-1$ first derivative constraints, for $2L$ constraints in total.
We do not justify why \eqref{interpolating_polynomial_expression} is the correct expression, leaving that to one of the many texts on the subject, such as \cite{burden2015numerical}.
Instead we will conclude by defining the divided difference, which enables evaluation of the above expression.

The divided difference of $f$ is defined recursively as follows.
First, to gain an intuition, if all of the $z_i$s were distinct then the divided difference would be defined by
\begin{equation}
\label{dd_recursion}
\begin{split}
    &f[z_i,z_{i-1},...,z_1,z_0]\\
    &\quad=\frac{f[z_i,z_{i-1},...,z_1]-f[z_{i-1},...,z_1,z_0]}{z_i-z_0},
\end{split}
\end{equation}
with the base case given by $f[z_j]=f(z_j)$.
Hence, one can think of the divided difference $f[z_i,z_{i-1},...,z_1,z_0]$ as analogous to an $i$th numerical derivative.

However, since in our case many adjacent pairs $z_{i+1}$ and $z_i$ are equal, the recursion above would become undefined when two arguments remain in the divided differences, e.g., $f[z_2,z_1]=\frac{f[z_2]-f[z_1]}{z_2-z_1}$ is undefined because $z_2=z_1$.
This is where the derivative constraints enter.
In our case, we define the base case at the level of two arguments as follows:
\begin{equation}
\label{dd_base_case_1}
    f[z_{i+1},z_i]=
    \begin{cases}
        \frac{f(z_{i+1})-f(z_i)}{z_{i+1}-z_i}\quad\text{if $z_{i+1}\neq z_i$},\\
        f'(z_i)\quad\text{if $z_{i+1}=z_i$}
    \end{cases}.
\end{equation}
In other words, since for example $f[z_2,z_1]=\frac{f[z_2]-f[z_1]}{z_2-z_1}$ is undefined because $z_2=z_1$, we replace it with the specified derivative at that point.
The recursion relation \eqref{dd_recursion} remains the same, but terminates at two arguments instead of one.

\eqref{dd_base_case_1} simplifies considerably.
Note that in terms of the $z_i$s, the constraints \eqref{constraints} become
\begin{equation}
    f(z_i)=f\left(x_{\left\lfloor\frac{i+1}{2}\right\rfloor}\right)=f_{\left\lfloor\frac{i+1}{2}\right\rfloor},
\end{equation}
and for $i=1,2,...,2L-2$,
\begin{equation}
    f'(z_i)=f'\left(x_{\left\lfloor\frac{i+1}{2}\right\rfloor}\right)=f'_{\left\lfloor\frac{i+1}{2}\right\rfloor}=0.
\end{equation}
Inserting these in \eqref{dd_base_case_1} yields
\begin{equation}
\label{dd_base_case_2}
    f[z_{i+1},z_i]
    =
    \begin{cases}
        \frac{f_{\left\lfloor\frac{i+2}{2}\right\rfloor}-f_{\left\lfloor\frac{i+1}{2}\right\rfloor}}{z_{i+1}-z_i}\quad\text{if $z_{i+1}\neq z_i$},\\
        0\quad\text{if $z_{i+1}=z_i$}
    \end{cases}.
\end{equation}
But most of the specified values are identical: we can see from \eqref{constraints} that $f_{\left\lfloor\frac{i+2}{2}\right\rfloor}-f_{\left\lfloor\frac{i+1}{2}\right\rfloor}=0$ unless
\begin{equation}
    \left\lfloor\frac{i+2}{2}\right\rfloor=\left\lceil\frac{L}{2}\right\rceil\quad\text{and}\quad\left\lfloor\frac{i+1}{2}\right\rfloor=\left\lfloor\frac{L}{2}\right\rfloor,
\end{equation}
which simplifies to
\begin{equation}
    i=L-1
\end{equation}
because $L$ is odd.
In this case, from \eqref{constraints} we see that $f_{\left\lfloor\frac{i+2}{2}\right\rfloor}-f_{\left\lfloor\frac{i+1}{2}\right\rfloor}=-2$, so \eqref{dd_base_case_2} becomes
\begin{equation}
\label{dd_base_case_3}
    f[z_{i+1},z_i]
    =\begin{cases}
        \frac{-2}{\cos\left(\frac{\pi}{L}\left\lceil\frac{L}{2}\right\rceil\right)-\cos\left(\frac{\pi}{L}\left\lfloor\frac{L}{2}\right\rfloor\right)}\quad\text{if $i=L-1$},\\
        0\quad\text{otherwise}.
    \end{cases}
\end{equation}
Combining this base case with the recursion \eqref{dd_recursion} yields all divided differences in the Hermite polynomial \eqref{interpolating_polynomial_expression}.

Note that the recursion relation \eqref{dd_recursion} might lead to a worry that evaluating the divided difference requires exponential time.
In fact, it can be compute efficiently as follows.
First, evaluate all of the two-argument divided differences as given by \eqref{dd_base_case_3}, of which there are $2L-1$, since there is one for each consecutive pair $z_{i+1},z_i$ and there are $2L$ $z_i$s.
Next, evaluate all of the three-argument divided differences, each of which is calculated by taking the difference of a consecutive pair of two-argument divided differences and dividing it by a difference between points, e.g.,
\begin{equation}
    f[z_{i+2},z_{i+1},z_i]=\frac{f[z_{i+2},z_{i+1}]-f[z_{i+1},z_i]}{z_{i+2}-z_i}.
\end{equation}
Hence, the number of three-argument divided differences is one fewer than the number of two-argument divided differences.
Then evaluate the four-argument divided differences using the three-argument divided differences in the same way, and so forth.

In this way, we build up a pyramid (called a \emph{divided differences table}) of all of the divided differences, where moving up the pyramid corresponds to divided differences with more arguments.
Since the base of the pyramid (the two-argument divided differences) has size $2L-1$, the total number of divided differences we need to evaluate to build up the whole pyramid is $O(L^2)$.

\end{document}